\documentclass[hidelinks,onefignum,onetabnum,nohypdvips]{siamart251216}

\newsiamremark{remark}{Remark}
\newsiamremark{hypothesis}{Hypothesis}
\crefname{hypothesis}{Hypothesis}{Hypotheses}
\newsiamthm{claim}{Claim}
\newsiamremark{fact}{Fact}
\crefname{fact}{Fact}{Facts}

\title{Superfast Spectral Compressed Sensing with Quadratic Convergence
}

\author{
	Xunmeng Wu,
	Jian-Feng Cai,
	Zongben Xu,
	Zai Yang
}

\usepackage{amsopn}
\DeclareMathOperator{\diag}{diag}

\usepackage{amsmath,amssymb,bm,cite,graphicx,algorithm,algorithmic}
\usepackage{epsfig,subfigure}
\usepackage{enumitem}
\usepackage{tablefootnote,makecell,multirow}
\usepackage{mathtools}

\usepackage{booktabs}
\usepackage{makecell}

\DeclareMathOperator*{\argmin}{arg\,min}

\DeclareMathOperator{\rank}{rank}
\def\norm #1{\left\|#1\right\|}

\def\twon #1{\left\|#1\right\|_2}

\def\frobn #1{\left\|#1\right\|_{\text{F}}}

\def\abs #1{\left|#1\right|}

\def\st{\text{subject to }}
\def\bC{\mathbb{C}}

\def\bR{\mathbb{R}}
\def\bE{\mathbb{E}}
\def\bP{\mathbb{P}}

\def\m #1{\boldsymbol{#1}}

\def\cC{\mathcal{C}}
\def\cD{\mathcal{D}}

\def\cG{\mathcal{G}}
\def\cH{\mathcal{H}}
\def\cI{\mathcal{I}}

\def\cM{\mathcal{M}}
\def\cN{\mathcal{N}}
\def\cO{\mathcal{O}}
\def\cP{\mathcal{P}}

\def\cT{\mathcal{T}}

\def\equ #1{\begin{equation}#1\end{equation}}   
\def\equs #1{\begin{equation*}#1\end{equation*}}  
\def\sbra #1{\left(#1\right)}
\def\mbra #1{\left[#1\right]}
\def\lbra #1{\left\{#1\right\}}
\def\diag #1{\text{diag}#1}

\def\st {\text{ subject to }}

\begin{document}

\maketitle

\footnotetext[1]{
	Department of Mathematics, Hong Kong University of Science and Technology,
	Kowloon, Hong Kong SAR, China.
	
	School of Mathematics and Statistics, Xi'an Jiaotong University,
	Xi'an, Shaanxi, China.
	
	E-mail:
	\email{xunmeng.wu@xjtu.edu.cn},
	\email{jfcai@ust.hk},
	\email{zbxu@xjtu.edu.cn},
	\email{yangzai@xjtu.edu.cn}.
}

\begin{abstract}
	In this paper, we study the spectral compressed sensing problem of recovering a spectrally sparse signal from linear compressive measurements. Exploiting the low-rank Hankel structure associated with the signal, we formulate a nonconvex fixed-rank optimization problem with a Hankel regularization.
	We propose a Superfast Iterative Hard Thresholding (SIHT) algorithm, where each iteration consists of a tangent-space least-squares update followed by a fixed-rank hard thresholding projection. 
	SIHT is the first nonconvex algorithm with provable quadratic convergence for spectral compressed sensing. For random Gaussian measurements, it achieves exact recovery in the noiseless case and stable recovery in the noisy case with near-optimal sample complexity.
	To efficiently implement SIHT, we develop a parameterized conjugate gradient method for solving the least-squares subproblem, endowing SIHT with state-of-the-art computational and memory efficiency. 
	Numerical experiments are provided to validate our results.
\end{abstract}



\section{Introduction}
Spectral compressed sensing aims to reconstruct a spectrally sparse signal from a small number of linear measurements \cite{duarte2013spectral,chen2014robust,cai2018spectral,candes2014towards,tang2013compressed}. A signal $\m{x}_\star \in \bC^N$ is called $r$-spectrally-sparse if it can be expressed as a superposition of $r$ complex exponentials:
\equs{
	x_{\star,j} = \sum_{k=1}^r s_k e^{(\mathrm{i} 2\pi f_k - \tau_k) j},
}
for $j = 0,1,\dots,N-1$, where $r \ll N$, $\mathrm{i}=\sqrt{-1}$ and $f_k \in [0,1)$, $\tau_k \in \bR$, and $s_k \in \bC$ denote the frequency, damping factor, and amplitude of the $k$-th spectral component, respectively. 
Given the sensing matrix $\m{A} \in \bC^{m\times N}$ with $m < N$ and noisy measurements
\equ{
	\m{y} = \m{A}\m{x}_\star + \m{\xi}, \label{eq:measurement}
}
where $\m{\xi}$ is the noise vector, the goal is to reconstruct $\m{x}_\star$. This problem arises in a broad range of applications, including wireless channel estimation \cite{heath2016overview}, sparse bandlimited signal acquisition \cite{tropp2009beyond}, array signal processing \cite{yang2018sparse}, radar sensing \cite{heckel2017generalized}, and nuclear magnetic resonance spectroscopy \cite{qu2015accelerated}.

Directly solving
\equs{
	\min_{\m{x}\in\bC^N} \twon{\m{A}\m{x}-\m{y}}^2, \quad \text{s.t. } \m{x} \text{ is $r$-spectrally sparse},
}
is challenging since the spectral sparsity constraint is implicitly defined and difficult to handle directly in the signal domain. As the spectrally sparse structure can be characterized by the low-rank Hankel structure in the lifted matrix domain, numerous methods based on low-rank Hankel model are proposed for spectral compressed sensing  \cite{chen2014robust,andersson2014new,cai2016robust,cai2018spectral,cai2019fast,yang2024new,cai2026spectral}.
Accordingly, this leads to a fixed-rank Hankel matrix recovery formulation
\equ{
	\min_{\m{x}} \twon{\m{A}\m{x}-\m{y}}^2, \quad \st \rank(\cH \m{x}) = r, \label{eq:P1}
}
where $\cH \m{x}$ is an $n_1 \times n_2$ Hankel matrix constructed from $\m{x}\in\bC^N$ with $N=n_1+n_2-1$, whose entries are given by
$\mbra{\cH\m{x}}_{ij}=x_{i+j}$ for $i=0,\dots,n_1-1$ and $j=0,\dots,n_2-1$.

Rather than directly solving problem \eqref{eq:P1}, early methods focus on its convex relaxation, namely Hankel nuclear norm minimization (Hankel NNM) \cite{cai2016robust,chen2014robust}. It is shown in \cite{cai2016robust} that Hankel NNM achieves exact noiseless recovery and stable noisy recovery under Gaussian measurements with high probability, requiring a sample complexity of $\cO\sbra{r\log^2 N}$. This rate is optimal with respect to $r$ since the degrees of freedom (DoF) of a spectrally sparse signal is $4r$. Another convex method is atomic norm minimization (ANM) \cite{tang2013compressed,candes2014towards,heckel2017generalized}, which achieves exact recovery from $\cO\sbra{r \log N}$ Gaussian measurements when the frequencies are sufficiently separated \cite{heckel2017generalized}. ANM has also been developed for multichannel frequency estimation using Hankel-Toeplitz block matrices \cite{wu2024multichannel}. However, both Hankel NNM and ANM are formulated as semidefinite programs and typically solved via interior-point methods (IPM) or the alternating direction method of multipliers (ADMM). This leads to a high computational complexity of $\cO(N^3)$, even when first-order ADMM solvers are employed.

To improve scalability, a fast iterative hard thresholding (FIHT) algorithm \cite{wu2026lifted,cai2019fast} is proposed for directly solving problem \eqref{eq:P1}. Compared with convex methods, FIHT enjoys a per-iteration cost that scales nearly linearly with $N$. As shown in \cite{wu2026lifted}, FIHT converges linearly to the true signal with high probability from $\cO\sbra{r\log r\log^2 N}$ Gaussian measurements, which is near-optimal in $r$ up to log factors.
Despite these advantages, FIHT still suffers from two limitations. First, its convergence rate is only linear, leaving room for further acceleration. Second, its theoretical guarantees in the noisy setting remain unclear. These limitations motivate the development of an algorithm that is simultaneously more efficient and provably robust to noise.

In this paper, we propose a Superfast Iterative Hard Thresholding (SIHT) algorithm, which, to the best of our knowledge, is the first nonconvex algorithm with provable quadratic convergence for spectral compressed sensing. Our main contributions are summarized as follows:

\begin{itemize}
	\item We reformulate \eqref{eq:P1} by incorporating the Hankel constraint into the objective via a Hankel regularization term, leading to an optimization problem on the fixed-rank manifold. We show that the proposed model admits order-optimal performance guarantees (see Section \ref{sec:formulation}).
	
	\item We develop the SIHT algorithm, where each iteration consists of two steps: solving a least-squares subproblem restricted to the tangent space of the fixed-rank manifold at the current iterate, followed by a hard-thresholding projection onto the manifold. As a byproduct, FIHT \cite{cai2019fast,wu2026lifted} can be interpreted as a special case of SIHT in which the least-squares subproblem is solved by only one gradient descent step (see Section \ref{sec:alg}).
	
	\item We establish quadratic convergence guarantees for SIHT under two different initialization schemes. SIHT achieves exact recovery in the noiseless setting and stable recovery in the noisy setting. Under Gaussian measurements, the required sample complexity with the resampling initialization is near-optimal up to log factors (see Section \ref{sec:theory}).
	
	\item To efficiently implement SIHT, we develop a parameterized conjugate gradient method to solve the least-squares subproblem. This yields per-iteration cost and memory usage that scale nearly linearly with the signal length $N$. For the second step, by exploiting the low-rank structure in the tangent space, the hard-thresholding projection can also be  computed efficiently with cost scaling nearly linearly with $N$ (see Section \ref{sec:implementation}).
	
\end{itemize}

Numerical experiments in Section \ref{sec:exp} demonstrate the efficiency and robustness of SIHT. The proof of the main result is presented in Section \ref{sec:recovery}, with additional technical details deferred to the appendices. The conclusion is presented in Section \ref{sec:conclusion}.

Table \ref{table:comparison} compares the sample and computational complexities of representative methods.
Compared with the ADMM implementations of ANM and Hankel NNM, SIHT reduces the per-iteration cost from cubic to nearly linear in $N$ and the memory requirement from quadratic to linear in $N$.
Compared with FIHT, SIHT improves the iteration complexity from log to double-log while retaining the same sample complexity and memory scaling.

\begin{table}[h!]
	\centering
	\renewcommand{\arraystretch}{1.25}
	\setlength{\tabcolsep}{3.5pt}
	
	\caption{
		Complexity comparison under noiseless Gaussian measurements, with $\cO\sbra{\cdot}$ omitted.
		ANM and Hankel NNM use ADMM, and ANM requires minimum frequency separation $\Delta_f \ge 4/N$.
		Here, $\kappa$ is the condition number of the true Hankel matrix, $K_t$ is the number of CG steps in outer iteration $t$.
	}
	\label{table:comparison}
	
		\begin{tabular}{@{}lcccc@{}}
		\toprule
		Method
		& \makecell{Sample\\complexity}
		& \makecell{Iteration\\complexity}
		& \makecell{Per-iteration\\cost}
		& Memory \\
		\midrule
		
		\makecell[l]{ANM \cite{heckel2017generalized}}
		& $r\log N$
		& $1/\varepsilon$
		& $N^3$
		& $N^2$ \\
		
		\addlinespace[3pt]
		
		\makecell[l]{Hankel NNM \cite{cai2016robust}}
		& $r\log^2 N$
		& $1/\varepsilon$
		& $N^3$
		& $N^2$ \\
		
		\addlinespace[3pt]
		
		FIHT \cite{wu2026lifted}
		& $r\log(r)\log^2(N)\kappa^2$
		& $\log\sbra{1/\varepsilon}$
		& \makecell{$mN+rN\log N$\\$+r^2N$}
		& $mN+rN$ \\
		
		\addlinespace[3pt]
		
		\makecell[l]{\textbf{SIHT}\\\textbf{(this work)}}
		& $r\log(r)\log^2(N)\kappa^2$
		& {\boldmath$\log\log\sbra{1/\varepsilon}$}
		& \makecell{$K_t(mN+rN\log N)$\\$+r^2N$}
		& $mN+rN$ \\
		
		\bottomrule
	\end{tabular}
	\end{table}

\emph{Notation:}
We denote by $\bR$ and $\bC$ the sets of real and complex numbers, respectively.
For a vector $\m{x}$, $\twon{\m{x}}$ and $\diag\sbra{\m{x}}$ denote its $\ell_2$-norm and the diagonal matrix with diagonal entries given by $\m{x}$, respectively.
For a matrix $\m{X}$, $\twon{\m{X}}$ and $\frobn{\m{X}}$ denote its spectral norm and Frobenius norm, respectively. For an operator $\cM$, we denote its adjoint as $\cM^*$ and define its operator norm as $\twon{\cM} := \sup_{\frobn{\m{X}}=1} \frobn{\cM\m{X}}$. The identity operator is denoted by $\mathcal{I}$.
We use $\overline{(\cdot)}$, $(\cdot)^T$, and $(\cdot)^H$ to denote complex conjugation, transpose, and conjugate transpose, respectively. $\mathcal{CN}(0,1)$ denotes the standard complex Gaussian distribution.
$\langle \cdot, \cdot \rangle$ denotes the standard inner product, and $\langle \cdot, \cdot \rangle_\bR$ denotes its real part. Finally, $m \gtrsim N$ means that $m \ge C N$ for some universal constant $C>0$.

\section{Optimization Model and Performance} \label{sec:formulation}
\subsection{Optimization Model}
Since the Hankel operator $\cH$ is not an isometry, as in \cite{cai2018spectral,cai2019fast,wu2026lifted}, we introduce a scaling operator $\cD := (\cH^* \cH)^{1/2}: \bC^N \to \bC^N$, and then define a scaled Hankel operator 
\equs{
	\cG := \cH \cD^{-1}: \bC^N \to \bC^{n_1 \times n_2}.
}
Since $\cG^*\cG = \cD^{-1} \cH^* \cH \cD^{-1} = \cI$, the operator $\cG$ is an isometry satisfying $\frobn{\cG \m{z}} = \twon{\m{z}}$ for any $\m{z}\in \bC^N$.

Let $\m{B} := \m{A} \cD^{-1}$ and $ \m{z}_\star := \cD\m{x}_\star$. Then the measurement model \eqref{eq:measurement} becomes
\equs{
	\m{y} = \m{A} \m{x}_\star + \m{\xi} = \m{B} \m{z}_\star + \m{\xi}.
}

Moreover, let $\m{z} := \cD \m{x}$. It follows that $\m{B}\m{z} = \m{A}\m{x}$ and $\cG\m{z} = \cH\m{x}$. Then, problem \eqref{eq:P1} is equivalent to
\equ{
	\min_{\m{z}\in \bC^N} \twon{\m{B}\m{z} - \m{y}}^2, 
	\quad \st \rank\sbra{\cG \m{z}} = r. 
	\label{eq:P2}
}
By separating the fixed-rank and Hankel constraints, problem \eqref{eq:P2} is rewritten as 
\equs{
	\min_{\m{z}\in \bC^N, \m{X}\in \bC^{n_1\times n_2} } \twon{\m{B}\m{z} - \m{y}}^2, 
	\quad \st \rank\sbra{\m{X}} = r, \quad  \m{X} = \cG\m{z}.
}
It follows from $\m{X} = \cG\m{z}$ that $\cG^*\m{X} = \cG^*\cG \m{z} = \m{z}$. Then substituting $\m{z} = \cG^*\m{X}$ into the above problem yields
\equs{
	\min_{\m{X}} \twon{\m{B}\cG^*\m{X}-\m{y}}^2, 
	\quad \st \rank\sbra{\m{X}} = r, \quad \sbra{\cI-\cG\cG^*}\m{X} = \m{0}.
}

Designing an algorithm that simultaneously enforces both the fixed-rank and Hankel constraints is challenging. To address this issue, we relax the Hankel constraint by introducing a quadratic regularization term:
\equ{
	\min_{\m{X}} \twon{\m{B}\cG^*\m{X}-\m{y}}^2 + \frobn{\sbra{\cI-\cG\cG^*}\m{X}}^2, \quad \st \rank\sbra{\m{X}} = r. \label{eq:P5}
}

\subsection{Model Performance}
We establish upper bounds on the signal recovery error for models \eqref{eq:P2} and \eqref{eq:P5}, and show that these bounds are tight up to constant factors.

To this end, we first assume that the sensing matrix $\m{B}$ satisfies the lifted restricted isometry property (LRIP) \cite{wu2026lifted}, which provides a fundamental tool for theoretical analysis of low-rank Hankel recovery problem. 

\begin{definition}[{\cite[Definition 1]{wu2026lifted}}] \label{def:LRIP}
	For an integer $r\ge1$, the LRIP constant $\delta_r$ of $\m{B}$ is defined as the smallest nonnegative number satisfying
	\equ{
		\abs{
			\twon{\m{B}\cG^*\m{Z}}^2
			-
			\twon{\cG^*\m{Z}}^2
		}
		\le
		\delta_r\frobn{\m{Z}}^2
		\label{eq:LRIP_matrix}
	}
	for all $\m{Z}\in\bC^{n_1\times n_2}$ with $\rank(\m{Z})\le r$.
	We say that $\m{B}$ satisfies the LRIP of rank $r$ if $\delta_r<1$.
\end{definition}

It is shown in \cite{wu2026lifted} that LRIP holds for Gaussian sensing matrices with a sufficient number of samples, as stated below.
\begin{theorem}[{\cite[Theorem 1]{wu2026lifted}}] \label{thm:LRIP}
	Let $0 \le \delta < 1$. Assume $\m{B}\in \bC^{m\times N}$ have entries that are i.i.d.\ $\cC\cN\sbra{0,1/m}$. Then, $\m{B}$ satisfies the LRIP with $\delta_r \le \delta$ with probability at least $1- N^{- (c_s + \log N)}$, provided $m \gtrsim C \delta^{-2}  r \sbra{ c_s + \log N } \log N$, where $c_s := n_2/n_1>1$.
\end{theorem}

Then, we establish the recovery error in the following theorem.
\begin{theorem} \label{thm:upperbound}
	Assume $\m{B}$ satisfies the LRIP of rank $2r$ with constant $\delta_{2r}\in [0,1)$. Let $\widehat{\m{z}}_0$ and $\widehat{\m{X}}$ be global minimizers of \eqref{eq:P2} and \eqref{eq:P5}, respectively. Define $\widehat{\m{z}}_1 = \cG^* \widehat{\m{X}}$. Then, for $i=0,1$, we have
	\equs{
		\twon{\widehat{\m{z}}_i - \m{z}_\star} \le \frac{2}{\sqrt{1-\delta_{2r}}} \twon{\m{\xi}}.
	}
\end{theorem}

\begin{proof}
	The proof is deferred to Appendix \ref{append:upperbound}.
\end{proof}

Theorem \ref{thm:upperbound} shows that the relaxed model \eqref{eq:P5} enjoys the same recovery guarantee as \eqref{eq:P2}.
In addition, a deterministic lower bound is provided in Theorem \ref{thm:lowerbound}. 
Together, these results show that the linear dependence of both recovery bounds on the noise norm is optimal in the deterministic worst case.

\begin{theorem} \label{thm:lowerbound}
	For any estimators $\widehat{\m{z}}$ constructed from $\m{y}$ and $\m{B}$, we have
	\equs{
		\inf_{\widehat{\m{z}}} \sup_{ \sbra{\m{B}, \m{z}_\star, \m{\xi}} \in F_r \sbra{\varepsilon} } \twon{\widehat{\m{z}} - \m{z}_\star} 
		\ge \frac{ \twon{\m \xi} }{\sqrt{2}},
	}
	where $F_r \sbra{\varepsilon} := \lbra{ \sbra{\m{B}, \m{z}_\star, \m{\xi}}: \text{$\m{B}$ satisfies LRIP, $\rank\sbra{\cG\m{z}_\star} \le r$, $\twon{\m \xi} = \varepsilon$} }$.
\end{theorem}
\begin{proof}
	The proof is deferred to Appendix \ref{append:lowerbound}.	
\end{proof}

\section{Algorithm} \label{sec:alg}
In this section, we propose an efficient algorithm to solve \eqref{eq:P5}.
\subsection{SIHT Algorithm} \label{subsec:SIHT}
Due to the fixed-rank constraint, problem \eqref{eq:P5} is nonconvex and NP-hard in general. Iterative hard thresholding (IHT) algorithms in compressed sensing \cite{blumensath2009iterative,foucart2011hard} provide a general framework for solving such fixed-cardinality problems and have been successfully extended to low-rank matrix recovery \cite{tanner2013normalized,wei2016guarantees}. IHT alternates between a descent step for reducing the objective and a hard thresholding step to enforce the constraint. 

Based on this framework, we propose a superfast IHT (SIHT) algorithm for solving \eqref{eq:P5}. Starting from the current iterate $\m{X}_t$ of rank $r$, each iteration consists of the following two steps. 

In the first step, inspired by Riemannian optimization \cite{absil2008optimization,luo2023low}, in \eqref{eq:P5} we approximate the rank-$r$ matrix manifold locally by its tangent space $T_t$ at $\m{X}_t$, which serves as a first-order approximation. This leads to a least squares problem constrained to $T_t$:
\equ{
	\m{W}_{t+1} = \argmin_{\m{Z}} \twon{\m{B}\cG^*\m{Z} - \m{y}}^2 + \frobn{\sbra{\cI-\cG\cG^*}\m{Z}}^2, \quad \st \m{Z}\in T_t.	\label{eq:subproblem}
}

In the second step, we project $\m{W}_{t+1} $ onto the rank-$r$ matrix manifold via the hard thresholding operator $\cT_r$, i.e., the best rank-$r$ approximation obtained by truncated singular value decomposition (SVD). 

The complete SIHT algorithm is summarized in Algorithm \ref{alg:SIHT}. The method is termed superfast as it achieves state-of-the-art efficiency in both convergence rate and per-iteration cost. Specifically, SIHT attains quadratic convergence, as established in Section \ref{sec:theory}, while each iteration can be implemented with cost scaling nearly linearly with $N$, as detailed in Section \ref{sec:implementation}.

\begin{algorithm}
	\caption{Superfast Iterative Hard Thresholding (SIHT) for solving \eqref{eq:P5}}
	\label{alg:SIHT}
	\begin{algorithmic}
		\REQUIRE Observations $\m{y}$, sensing matrix $\m{B}$, and sparsity level $r$ \\
		\STATE Initialize $\m{X}_0$
		\FOR{$t = 0,1,...$}
		\STATE \textbf{Step 1:} Obtain $\m{W}_{t+1}$ by solving \eqref{eq:subproblem} 
		\STATE \textbf{Step 2:} Set $\m{X}_{t+1} = \cT_r(\m{W}_{t+1})$
		\ENDFOR
		\ENSURE $\m{X}_{t+1}$
	\end{algorithmic}
\end{algorithm}

\begin{remark}
	The subproblem \eqref{eq:subproblem} provides insights for understanding the FIHT algorithm \cite{cai2019fast,wu2026lifted} and its difference from the proposed SIHT. Specifically, the Riemannian gradient of \eqref{eq:subproblem} is given by $\cP_{T_t} \sbra{ \cG\m{B}^H \sbra{\m{B}\cG^* \m{Z} -\m{y}} + \sbra{\cI-\cG\cG^*}\m{Z} }$. Starting from the current iterate $\m{X}_t \in T_t$, performing one step Riemannian gradient descent with unit stepsize yields
	\equs{
		\begin{split}
			\breve{\m{W}}_{t+1}
			& = \m{X}_t - \cP_{T_t} \sbra{ \cG\m{B}^H \sbra{\m{B}\cG^* \m{X}_t -\m{y}} + \sbra{\cI-\cG\cG^*}\m{X}_t } \\
			& = \m{X}_t - \cP_{T_t} \sbra{ \cG\m{B}^H \sbra{\m{B} \m{z}_t -\m{y}} + \m{X}_t - \cG \m{z}_t } \\
			& = \cP_{T_t} \cG \sbra{ \m{z}_t - \m{B}^H \sbra{\m{B}\m{z}_t - \m{y} } },
		\end{split}
	}
	where $\m{z}_t = \cG^* \m{X}_t$. This update coincides with the core step of FIHT.
	
	This observation shows that FIHT can be interpreted as performing one Riemannian gradient step for solving \eqref{eq:subproblem}, whereas SIHT solves \eqref{eq:subproblem} exactly. As a result, SIHT achieves quadratic convergence, while FIHT only has linear convergence.
\end{remark}

\subsection{Initialization} \label{subsec:init}
We present two initialization schemes for $\m{X}_0$ in SIHT.

The first scheme is one-step hard thresholding, defined by $\m{X}_0 = \cT_r\bigl( \cG \m{B}^H \m{y} \bigr)$.
This scheme is simple to implement, but its recovery guarantee generally requires a suboptimal sample complexity that scales quadratically with the rank $r$; see \eqref{eq:m_guarantee1}.

To obtain stronger recovery guarantees, we further consider a resampling-based initialization scheme and set $\m{X}_0=\tilde{\m{X}}_T$, where $\tilde{\m{X}}_T$ denotes the output of a resampled FIHT algorithm \cite{wu2026lifted}; see Appendix \ref{append:FIHT}. The key benefit of resampling is that each iterate $\tilde{\m{X}}_t$ is statistically independent of the fresh measurements used in the next iteration. This independence enables SIHT to achieve a near-optimal sample complexity, scaling linearly with $r$ up to log factors; see \eqref{eq:m_guarantee2}.

\section{Recovery Guarantee} \label{sec:theory}
We have the following theorem for SIHT.

\begin{theorem} \label{thm:recovery}
	Denote the condition number $\kappa := \sigma_{\max} \sbra{\cG\m{z}_\star} / \sigma_{\min} \sbra{\cG\m{z}_\star}$ and  $c_s := n_2 / n_1$. Assume $\m{B} \in \bC^{m\times N}$ has i.i.d.\ entries following $\cC\cN\sbra{0,1/m}$ and $\m{\xi} \in \bC^m$ has i.i.d.\ entries following $\cC\cN(0,\sigma^2/m)$. Let $\m{z}_t = \cG^* \m{X}_t$, where $\m{X}_t$ is the SIHT iterate. 
	Then,
	\equ{
		\twon{\m{z}_t - \m{z}_\star} 
		\le 2^{-2^t} \sigma_{\min} \sbra{\cG\m{z}_\star} 
		+ \sigma \sqrt{ \frac{ Cr(c_s + \log N)\log N} {m} }, 
		\quad \forall t \ge 0,
		\label{eq:main_recovery}
	}
	for some universal constant $C>0$,
	under either of the following initializations.
	
	(i) One-step initialization:
	Let $
	\m{X}_0
	=
	\cT_r\sbra{\cG\m{B}^H\m{y}}$,
	then with probability at least
	$1-2N^{-(c_s+\log N)}$, \eqref{eq:main_recovery} holds provided
	\equ{
		m \gtrsim \kappa^2 r \sbra{ r + \frac{\sigma^2}{\sigma_{\max}^2\sbra{\cG\m{z}_\star}} } 
		\max\lbra{ c_s, \log N } \log N.
		\label{eq:m_guarantee1}
	}
	
	(ii) Resampling initialization:
	Let
	$\m{X}_0 = \tilde{\m{X}}_T$
	with
	$T\ge\lceil \log_2 r \rceil$ where $\tilde{\m{X}}_T$ is the output of Algorithm \ref{alg:Init-resamp}, then with probability exceeding
	$1-(4+2T)N^{-(c_s+\log N)}$, \eqref{eq:main_recovery} holds provided
	\equ{
		m \gtrsim \kappa^2 \log \sbra{r} \sbra{ r + \frac{\sigma^2}{\sigma_{\max}^2\sbra{\cG\m{z}_\star}} } 
		\max\lbra{ c_s, \log N } \log N.
		\label{eq:m_guarantee2}
	}
\end{theorem}

\begin{proof}
	The proof is deferred to Section \ref{sec:recovery}.
\end{proof}

Theorem \ref{thm:recovery} reveals several notable properties of SIHT.

\paragraph{Quadratic convergence}
In the noiseless setting where $\m{\xi}=\m{0}$, SIHT achieves an $\varepsilon$-accurate solution satisfying $\twon{\m{z}_t-\m{z}_\star} \le \varepsilon \sigma_{\min}\sbra{\cG\m{z}_\star}$ within
$\cO\bigl(\log\log(1/\varepsilon)\bigr)$
iterations. Hence, SIHT enjoys quadratic convergence. This significantly improves over FIHT \cite{cai2019fast,wu2026lifted}, which requires
$\cO\bigl(\log(1/\varepsilon)\bigr)$
iterations and only achieves linear convergence.

\paragraph{Noise robustness}
SIHT achieves stable recovery under Gaussian noise.
In \eqref{eq:main_recovery}, the first term decays doubly exponentially with the iterations.
Thus, SIHT rapidly achieves a recovery error bounded by a constant multiple of $\sigma$.

\paragraph{Near-optimal sample complexity}
In the low-noise regime $\sigma \ll \sigma_{\max}(\mathcal G z_\star)$, the sample complexity in \eqref{eq:m_guarantee2} reduces to $m \gtrsim \kappa^2 r \log(r) \max\lbra{c_s,\log N}\log N$, matching that of FIHT \cite{wu2026lifted}. Since the DoF of the ground-truth signal $\m{z}_\star$ are only $4r$, the sample complexity is near-optimal up to conditioning and log factors.

\section{Efficient Implementation} \label{sec:implementation}

\subsection{Step 1: Parameterized Conjugate Gradient (CG) Algorithm for solving (\ref{eq:subproblem})} \label{subsubsec:CG}
\paragraph{Preliminaries}
Suppose the rank-$r$ matrix $\m{X}_t$ admits the SVD $\m{X}_t = \m{U}_t \m{\Sigma}_t \m{V}_t^H$, where $\m{U}_t \in \bC^{n_1 \times r}$, $\m{\Sigma}_t = \diag\sbra{\m{\sigma}_t}$, and $\m{V}_t \in \bC^{n_2 \times r}$. Then, the tangent space at $\m{X}_t$ is given by
\equ{
	T_t := \lbra{ \m{U}_t \m{C}^H_2 + \m{C}_1 \m{V}_t^H : 
		\m{C}_1 \in \bC^{n_1 \times r},\; \m{C}_2 \in \bC^{n_2 \times r} },
	\label{eq:Tangentspace}
}
and the orthogonal projection of $\m{Z}\in \bC^{n_1 \times n_2}$ onto $T_t$ is given by
\equs{
	\cP_{T_t}(\m{Z})
	= \m{U}_t\m{U}_t^H \m{Z} + \m{Z}\m{V}_t\m{V}_t^H - \m{U}_t\m{U}_t^H \m{Z}\m{V}_t\m{V}_t^H.
}
For notational simplicity, we omit the subscript $t$ in the following.

\paragraph{Formulation}
From the Riemannian first-order optimality condition of problem \eqref{eq:subproblem}, we obtain
\equ{
	\begin{split}
		\cP_T\sbra{\cG\sbra{\m{B}^H\m{B} - \m{I}}\cG^* + \cI } \sbra{ \m{Z} } 
		= \cP_T \cG \m{B}^H \m{y}, \qquad\m{Z} \in T. 
	\end{split} \label{eq:linearEq}
}
Denote 
\equs{
	\cM:=\cP_T\sbra{\cG\sbra{\m{B}^H\m{B} - \m{I}}\cG^* + \cI }\cP_T, \qquad \m{D}:=\cP_T \cG \m{B}^H \m{y}.
}
Then solving \eqref{eq:subproblem} is equivalent to solving the linear system 
\equ{
	\cM\sbra{\m{Z}}=\m{D}, \qquad \m{Z} \in T. \label{eq:AXD}
}

We assume that the self-adjoint operator $\cM$ is invertible on the tangent space $T$, which will be established in Lemma \ref{lem:PTtPTt}.
Then the solution can be formally written as $\m{Z}=\cM^{-1} \m{D}$. But explicitly forming or inverting $\cM$ is computationally expensive due to the high dimension of $T$. Instead, we are to exploit the structure of $\cM$ and solve \eqref{eq:AXD} efficiently using the CG algorithm.

We point out that, a direct application of CG requires storing three full matrices of size $n_1 \times n_2$, namely the iterate $\m{Z}$, the residual $\m{R}$, and the search direction $\m{P}$, which is not scalable to large-scale problems. In addition, Step 2 requires performing a truncated SVD on the output $\m{W}_{t+1}$ from Step 1, which suffers from high computational cost.

\paragraph{Parameterization}
To address aforementioned issues, we exploit the low-rank and Hankel structure of the iterate via parameterization. Specifically, we parameterize $\m{Z} \in T$ as
\equs{
	\m{Z} = f(\m{q}^Z, \m{\sigma}^Z) := \cP_T \cG \m{q}^Z + \m{U} \m{\Sigma}^Z \m{V}^H,
}
where $\m{q}^Z \in \bC^N$ and $\m{\Sigma}^Z = \diag(\m{\sigma}^Z) \in \bR^{r \times r}$ with $\m{\sigma}^Z \in \bR^r$. 
Similarly, the residual $\m{R}$ and search direction $\m{P}$ admit the parameterizations
\equs{
	\begin{split}
		\m{R} & = f(\m{q}^R, \m{\sigma}^R) := \cP_T \cG \m{q}^R + \m{U} \m{\Sigma}^R \m{V}^H, \\ 	
		\m{P} & = f(\m{q}^P, \m{\sigma}^P) := \cP_T \cG \m{q}^P + \m{U} \m{\Sigma}^P \m{V}^H,
	\end{split}
}
where $\m{q}^R, \m{q}^P \in \bC^N$, $\m{\Sigma}^R = \diag(\m{\sigma}^R) $ and $\m{\Sigma}^P = \diag(\m{\sigma}^P) $ with $\m{\sigma}^R,\m{\sigma}^P \in \bR^r$.

Under this parameterization, each variable (e.g., $\m{Z}$, $\m{R}$, and $\m{P}$) is represented by a pair $(\m{q}, \m{\sigma}) \in \bC^N \oplus\bR^r$, which has $2N + r$ real DoF. Consequently, the CG iterations can be carried out entirely in this reduced parameter space, avoiding the explicit storage of full matrices in $\bC^{n_1 \times n_2}$ (note that $r \ll N = n_1 + n_2 - 1$). Moreover, it eliminates the need for truncated SVD of the full matrix $\m{W}_{t+1}$, as detailed in Section \ref{subsubsec:step2}.

\begin{algorithm}
	\caption{Parameterized CG for solving problem \eqref{eq:AXD}}
	\label{alg:CG}
	\begin{algorithmic}
		\REQUIRE $\m{y}$, $\m{B}$, $(\m{U}, \m{\Sigma}, \m{V})$ of $\m{X}_t$ 
		\STATE Initialize $\lbra{\m{q}_0^Z, \m{\sigma}_0^Z}$, $\lbra{\m{q}_0^R, \m{\sigma}_0^R}$, and $\lbra{\m{q}_0^P, \m{\sigma}_0^P}$ via \eqref{eq:init_XRP}
		
		\FOR{$k = 0,1,\dots$}
		
		\STATE Compute $\alpha_k = \left \langle \m{R}_k, \m{R}_k \right \rangle / \left \langle \m{P}_k, \cM\sbra{\m{P}_k} \right \rangle$ via \eqref{eq:RR} and \eqref{eq:PAP}
		
		\STATE Update $\m{q}^Z_{k+1} = \m{q}^Z_k + \alpha_k \m{q}^P_k$ and $\m{\sigma}_{k+1}^Z = \m{\sigma}_k^Z + \alpha_k \m{\sigma}_k^P$
		
		\STATE Update $\m{q}^R_{k+1} = \m{q}_k^R - \alpha_k \m{q}^\circ_k$ and $\m{\sigma}_{k+1}^R = \m{\sigma}_k^R - \alpha_k \m{\sigma}^P_k$
		
		\IF{$\frobn{\m{R}_{k+1}} < \epsilon$}
		\STATE break
		\ENDIF
		
		\STATE Update $\beta_k = \left \langle \m{R}_{k+1}, \m{R}_{k+1} \right \rangle / \left \langle \m{R}_k, \m{R}_k \right \rangle$
		
		\STATE Update $\m{q}^P_{k+1} = \m{q}^R_{k+1} + \beta_k \m{q}^P_k$ and $\m{\sigma}^P_{k+1} = \m{\sigma}_{k+1}^R + \beta_k \m{\sigma}_k^P$
		\ENDFOR
		
		\ENSURE $\lbra{\m{q}^Z_{k+1}, \m{\sigma}_{k+1}^Z}$
	\end{algorithmic}
\end{algorithm}

\paragraph{Parameterized CG}
Starting from the standard CG updates, we derive the corresponding expressions under our parameterization as follows, and summarize the resulting algorithm in Algorithm \ref{alg:CG}.
	\subparagraph{Step 1.1 Initialization} Note that $\cP_T \sbra{\m{X}_t} = \m{X}_t$. We have
	\equ{
		\begin{split}
			\m{Z}_0 = \m{X}_t = \m{U} \m{\Sigma} \m{V}^H  \quad & \Leftrightarrow \quad \m{q}^Z_0 = \m{0}, \ \m{\sigma}^Z_0 = \m{\sigma}, \\
			\m{R}_0 = \m{D} - \cM\sbra{\m{Z}_0} \quad & \Leftrightarrow \quad \m{q}^R_0 = \m{B}^H \sbra{\m{y} - \m{B}\cG^*\sbra{\m{U}\m{\Sigma}\m{V}^H} } + \cG^*\sbra{\m{U}\m{\Sigma}\m{V}^H}, \\ 
			& \quad \quad \ \; \m{\sigma}^R_0 = - \m{\sigma}, \\
			\m{P}_0 = \m{R}_0 \quad & \Leftrightarrow \quad \m{q}^P_0 = \m{q}^R_0, \ \m{\sigma}^P_0 = \m{\sigma}^R_0.
		\end{split} \label{eq:init_XRP}
	}
	\subparagraph{Step 1.2 Compute stepsize} $\alpha_k = \left \langle \m{R}_k, \m{R}_k \right \rangle / \left \langle \m{P}_k, \cM\sbra{\m{P}_k} \right \rangle$. Specifically, the numerator admits
	\equ{
		\begin{split}
			\left \langle \m{R}_k, \m{R}_k \right \rangle 
			& = \left \langle \cP_T \cG \m{q}^R_k + \m{U} \m{\Sigma}^R_k \m{V}^H, \cP_T \cG \m{q}^R_k + \m{U} \m{\Sigma}^R_k \m{V}^H \right \rangle  \\
			& = \left \langle \m{q}^R_k, \cG^* \cP_T \cG \m{q}_k^R \right \rangle + 2 \left \langle \m{q}_k^R, \cG^* \sbra{ \m{U}\m{\Sigma}^R_k \m{V}^H } \right \rangle_{\bR} + \frobn{ \m{\Sigma}_k^R }^2.
		\end{split} \label{eq:RR}
	}
	Since $\cM\sbra{\m{P}_k}$ can be parameterized as $\cM\sbra{\m{P}_k} = \cP_T \cG \m{q}^\circ_k  + \m{U} \m{\Sigma}^P_k \m{V}^H$,
	where $$
		\m{q}^\circ_k := \sbra{\m{B}^H\m{B} - \m{I}}\cG^* \sbra{ \cP_T \cG \m{q}^P_k + \m{U} \m{\Sigma}^P_k \m{V}^H } + \m{q}^P_k,$$
	the denominator admits
	\equ{
		\begin{split}
			\left \langle \m{P}_k, \cM\sbra{\m{P}_k} \right \rangle 
			& = \left \langle \cP_T \cG \m{q}^P_k + \m{U} \m{\Sigma}^P_k \m{V}^H, \cP_T \cG \m{q}^\circ_k  + \m{U} \m{\Sigma}^P_k \m{V}^H \right \rangle  \\
			& = \left \langle \cG^* \cP_T \cG \m{q}^P_k, \m{q}_k^\circ \right \rangle 
			+ \left \langle \m{q}_k^P, \cG^* \sbra{ \m{U}\m{\Sigma}^P_k \m{V}^H } \right \rangle \\
			& \quad + \left \langle \cG^* \sbra{ \m{U}\m{\Sigma}^P_k \m{V}^H }, \m{q}_k^\circ \right \rangle
			+ \frobn{\m{\Sigma}_k^P}^2.
		\end{split} \label{eq:PAP}
	}
	
	\subparagraph{Step 1.3 Update iterate}
	\equs{
		\m{Z}_{k+1} = \m{Z}_k + \alpha_k \m{P}_k 
		\quad \Leftrightarrow \quad
		\m{q}^Z_{k+1} = \m{q}^Z_k + \alpha_k \m{q}^P_k, \quad \m{\sigma}_{k+1}^Z = \m{\sigma}_k^Z + \alpha_k \m{\sigma}_k^P. 
	}
	
	\subparagraph{Step 1.4 Update residual}
	\equs{
		\m{R}_{k+1} = \m{R}_k - \alpha_k \cM \sbra{\m{P}_k} 
		\quad \Leftrightarrow \quad 
		\m{q}^R_{k+1} = \m{q}_k^R - \alpha_k \m{q}^\circ, \quad \m{\sigma}_{k+1}^R = \m{\sigma}_k^R - \alpha_k \m{\sigma}^P_k.
	}
	\subparagraph{Step 1.5 Check convergence} Compute $\frobn{\m{R}_{k+1}}^2 = \left \langle \m{R}_{k+1}, \m{R}_{k+1} \right \rangle$ as in \eqref{eq:RR}.
	\subparagraph{Step 1.6 Compute conjugate coefficient} Compute $\beta_k = \left \langle \m{R}_{k+1}, \m{R}_{k+1} \right \rangle / \left \langle \m{R}_k, \m{R}_k \right \rangle$.
	\subparagraph{Step 1.7 Update search direction} 
	\equs{
		\m{P}_{k+1} = \m{R}_{k+1} + \beta_k \m{P}_k 
		\quad \Leftrightarrow \quad 
		\m{q}^P_{k+1} = \m{q}^R_{k+1} + \beta_k \m{q}^P_k, \quad \m{\sigma}^P_{k+1} = \m{\sigma}_{k+1}^R + \beta_k \m{\sigma}_k^P.
	}

The parameterized CG updates involve two types of operations:
\begin{itemize}
	\item the adjoint Hankel operator applied to the product of rank-$r$ factors, i.e., $\cG^*(\m{U}\m{\Sigma}\m{V}^H)$;
	\item multiplications between Hankel matrix $\cG\m{q}$ and rank-$r$ factors in $\cP_T\sbra{\cdot}$, i.e., $\m{U}^H(\cG \m{q})$ and $(\cG \m{q})\m{V}$.
\end{itemize}
As shown in \cite[Section 2.2]{cai2019fast}, these operations can be efficiently implemented via the fast Fourier transform (FFT), each requiring $\cO(r N \log N)$ flops. Additionally, computing $\m{B}\m{x}$ for a vector $\m{x}$ incurs $\cO(mN)$ flops. Therefore, the overall computational cost per CG iteration is $\cO(mN + r N \log N)$. 

As will be shown in Lemma \ref{lem:PTtPTt}, $\cM$ is positive definite on the tangent space $T$ with a bounded condition number. It then follows from \cite[Section 5.1.5]{conn2000trust} that, under exact arithmetic, CG solves the linear system \eqref{eq:AXD} exactly in at most $r(n_1+n_2-r)$ iterations, the complex dimension of $T$. Its bounded condition number ensures a dimension-independent linear convergence rate.

\subsection{Step 2: Efficient SVD via Low-Rank Structure} \label{subsubsec:step2}
Note that $\m{\Sigma}_{k+1}^Z = \diag\sbra{\m{\sigma}_{k+1}^Z}$. Define
\equs{
	\begin{split}
		\m{C}&:=\m{U}^H \cG \m{q}^Z_{k+1} \m{V} \in \bC^{r \times r}, \\
		\m{E}&:= (\m{I} - \m{V} \m{V}^H)(\cG \m{q}^Z_{k+1})^H \m{U} \in \bC^{n_2 \times r}, \\ 
		\m{F}&:=(\m{I} - \m{U} \m{U}^H ) \cG \m{q}^Z_{k+1} \m{V} \in \bC^{n_1 \times r}.
	\end{split}
}
Let $\m{E} = \m{Q}_1\m{R}_1$ and $\m{F} = \m{Q}_2 \m{R}_2$ be the QR factorizations, where $\m{Q}_1 \in \bC^{n_2\times r}$ and $\m{Q}_2 \in \bC^{n_1\times r}$ satisfying $\m{Q}_1^H \m{Q}_1 = \m{I}$ and $\m{Q}_2^H \m{Q}_2 = \m{I}$, and $\m{R}_1,\m{R}_2 \in \bC^{r \times r}$.

For Step 2 in SIHT, instead of explicitly forming $\m{W}_{t+1}$, we work with its parameterized form:
\equ{
	\begin{split}
		& \cP_T \cG \m{q}^Z_{k+1} + \m{U} \m{\Sigma}_{k+1}^Z \m{V}^H \\
		& = \m{U} \m{U}^H \cG \m{q}^Z_{k+1} + \cG \m{q}^Z_{k+1} \m{V} \m{V}^H - \m{U} \m{U}^H \cG \m{q}^Z_{k+1} \m{V} \m{V}^H + \m{U} \m{\Sigma}_{k+1}^Z \m{V}^H \\
		& = \m{U}  \m{U}^H \cG \m{q}^Z_{k+1} \sbra{ \m{I} - \m{V} \m{V}^H}  + \sbra{ \m{I} - \m{U} \m{U}^H } \cG \m{q}^Z_{k+1} \m{V}  \m{V}^H \\
		& \quad + \m{U}  \m{U}^H \cG \m{q}^Z_{k+1} \m{V}  \m{V}^H + \m{U} \m{\Sigma}_{k+1}^Z \m{V}^H  \\
		& = \m{U} \m{E}^H + \m{F} \m{V}^H + \m{U} \sbra{\m{C}+\m{\Sigma}_{k+1}^Z} \m{V}^H  \\
		& = \begin{bmatrix} \m{U} & \m{Q}_2 \end{bmatrix} \begin{bmatrix} \m{C}+\m{\Sigma}_{k+1}^Z & \m{R}_1^H \\ \m{R}_2 & \m{0} \end{bmatrix} \begin{bmatrix} \m{V} & \m{Q}_1 \end{bmatrix}^H.
	\end{split} \label{eq:Wt+1P}
}
Let the SVD of the middle $2r \times 2r$ matrix in the last line of \eqref{eq:Wt+1P} be $ \m{U}_c \m{\Sigma}_c \m{V}_c^H$.
Then 
$
	\sbra{ \begin{bmatrix} \m{U} & \m{Q}_2 \end{bmatrix} \m{U}_c } \m{\Sigma}_c \sbra{ \begin{bmatrix} \m{V} & \m{Q}_1 \end{bmatrix} \m{V}_c }^H
$
is a SVD of $\m{W}_{t+1}$, since $\m{V}^H \m{Q}_1 = \m{0}$ and $\m{U}^H \m{Q}_2 = \m{0}$. Therefore, $\m{X}_{t+1} = \cT_r(\m{W}_{t+1})$ can be obtained by retaining the leading $r$ singular components.

The computational cost is dominated by forming $\m{C}$, $\m{E}$ and $\m{F}$, which requires $\cO(r^2 N + r N \log N)$ flops, and performing QR factorizations of $\m{E}$ and $\m{F}$, which requires $\cO(r^2 N)$ flops, and computing the SVD of the $2r \times 2r$ matrix, which costs $\cO(r^3)$ flops.

\subsection{Complexity of SIHT}
In summary, with a fixed maximum number of CG iterations, the computational cost per outer iteration of SIHT is $\mathcal{O}(mN+rN\log N+r^2N+r^3)$. The memory usage is dominated by storing the sensing matrix $\m{B} \in \bC^{m \times N}$ and the matrices $\m{U},\m{Q}_2 \in \bC^{n_1 \times r}$ and $\m{V}, \m{Q}_1 \in \bC^{n_2 \times r}$, leading to an overall memory complexity of $\cO\sbra{mN + rN}$.

\section{Numerical Experiments} \label{sec:exp}
We evaluate the empirical performance of SIHT and compare it with FIHT \cite{wu2026lifted}, Hankel-NNM \cite{cai2016robust}, and ANM \cite{heckel2017generalized}. For SIHT, the outer iterations are terminated when the relative change $\norm{\m{z}_{t+1} - \m{z}_t }_2 / \norm{\m{z}_t}_2 \le 10^{-8}$ where $\m{z}_t = \cG^* \m{X}_t$ or when the maximum number of iterations $T_{max}$ is reached. In the inner loop, corresponding to the parameterized CG iterations, we set the maximum number of iterations to $K_{max}$ and the residual tolerance to $\epsilon = 10^{-8}$. SIHT is initialized using one-step hard thresholding rather than the resampling scheme, as preliminary simulations indicate comparable performance between the two schemes while the former is computationally simpler. FIHT uses the same initialization and stopping criteria as SIHT. For Hankel-NNM and ANM, we employ the interior-point solver SDPT3 via the CVX toolbox \cite{cvx}. Throughout this section, the signal recovery accuracy is measured by the relative root-mean-squared error (RMSE) $\twon{\widehat{\m{x}} - \m{x}_\star} / \twon{\m{x}_\star}$, where $\widehat{\m{x}}$ denotes the recovered signal. 

\begin{figure*}[htb] 
	\centering
	\subfigure {\includegraphics[width=6.0cm]{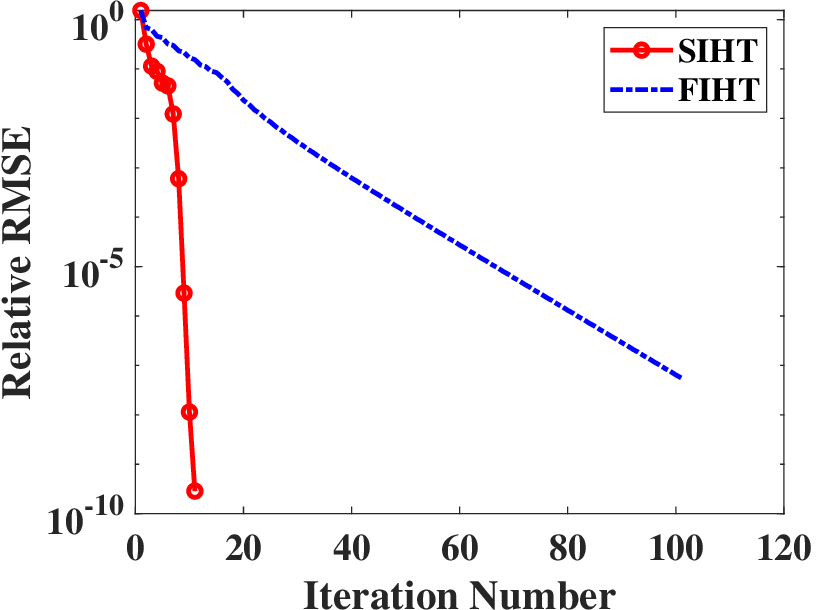} }
	\subfigure {\includegraphics[width=6.0cm]{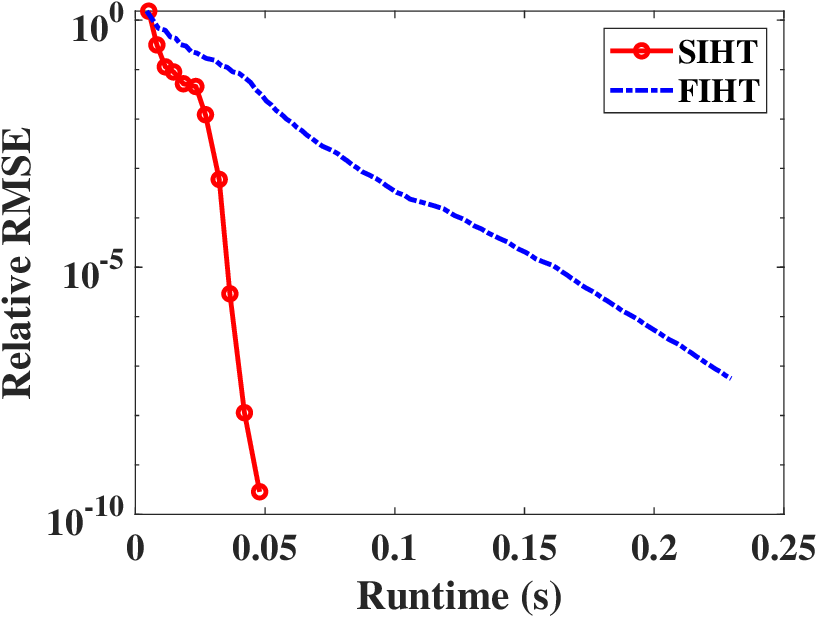} }
	
	\subfigure {\includegraphics[width=6.0cm]{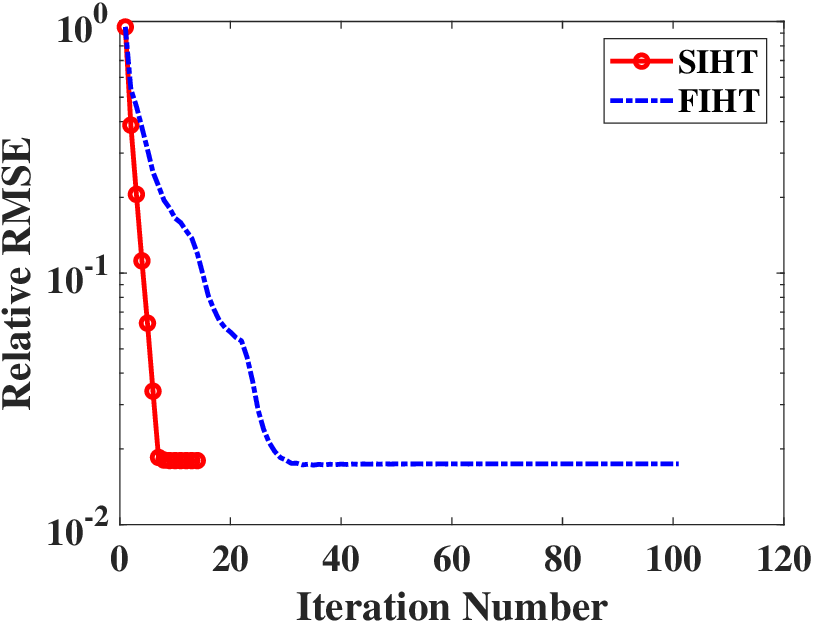} }
	\subfigure {\includegraphics[width=6.0cm]{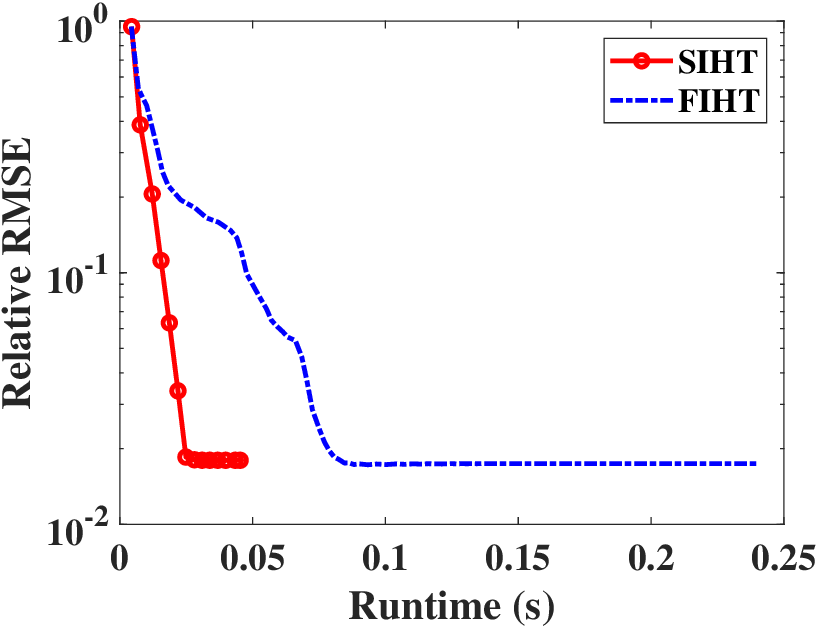} }
	\caption{Relative RMSE versus iteration for the proposed SIHT and FIHT \cite{wu2026lifted}. \emph{Top}: noiseless case. \emph{Bottom}: noisy case with $\mathrm{SNR}=20\,\mathrm{dB}$. Here, $N=127$, $m=90$, and $r=15$.}
	\label{fig:time}
\end{figure*}

\subsection{Convergence Performance}
We generate spectrally sparse signals $\m{x}_\star$ of length $N=127$ with $r=15$ frequency components. Each frequency $f_k$ is sampled uniformly from $[0,1)$ and the damping factors are set to $\tau_k=0$. The amplitudes are generated as $s_k = \sbra{1+10^{0.5 m_k}}e^{i 2\pi \phi_k}$, where $m_k$ and $\phi_k$ are independently drawn from the uniform distribution on $[0,1]$. The Hankel matrix dimension is chosen as $n_1=n_2=64$.
We acquire $m=90$ measurements of $\m{x}_\star$ through random Gaussian sampling, i.e., $\m{y}=\m{B}\cD\m{x}_\star$, where the entries of $\m{B}$ follow $\cC\cN(0,1/m)$. In the noisy case, complex white Gaussian noise $\m{\xi}$ is added to the measurements $\m{y}$. The signal-to-noise ratio (SNR) is defined as $\text{SNR} =10\log_{10} (\twon{ \m{y} }^2/\twon{\m{\xi}}^2)$. For SIHT, we set $T_{max}=100$ and $K_{max} = 10$. For FIHT, we set $T_{max}=100$.

Figure \ref{fig:time} presents the relative RMSE versus iteration number and runtime. SIHT exhibits a quadratic convergence to machine precision in the noiseless case and to an error floor determined by the noise level in the noisy case, while FIHT only exhibits a linear convergence. This is consistent with the predicted convergence behavior in Section \ref{sec:theory}. Moreover, SIHT requires substantially less runtime than FIHT to achieve the same RMSE, thanks to the efficient implementation in Section \ref{sec:implementation}.

\subsection{Phase Transition}
A signal is regarded as successfully recovered if the relative RMSE is below $10^{-3}$. For each parameter pair $(m,r)$, the success rate is estimated over $20$ independent Monte Carlo trials, where $m\in\lbra{20,24,\ldots,124}$ and $r\in\lbra{1,3,\ldots,39}$. For SIHT, we set $T_{max}=200$ and $K_{max} = 50$. For FIHT, we set $T_{max}=1000$.
Figure \ref{fig:exp1} illustrates the resulting phase transition results in the noiseless case. SIHT achieves slightly better recovery performance than FIHT and Hankel-NNM, while substantially outperforming ANM, which performs worse due to closely spaced frequencies.

\begin{figure*}[t] 
	\begin{center}
		\subfigure[Proposed SIHT]%
		{\includegraphics[width=.40\textwidth]{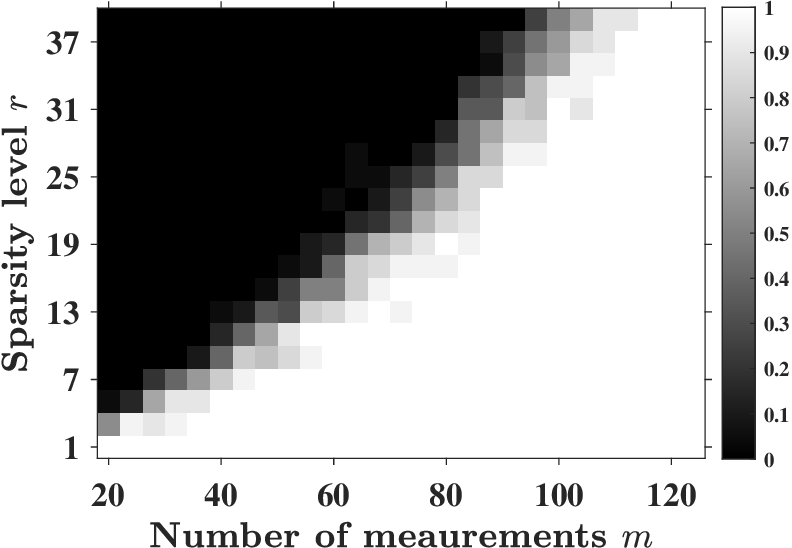}}
		\qquad
		\subfigure[FIHT \cite{wu2026lifted}]%
		{\includegraphics[width=.40\textwidth]{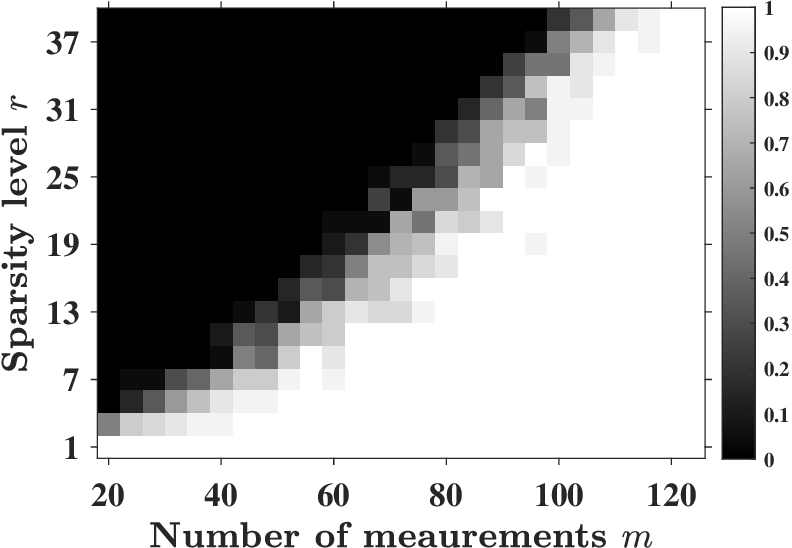}}
		\subfigure[Hankel-NNM \cite{cai2016robust}]%
		{\includegraphics[width=.40\textwidth]{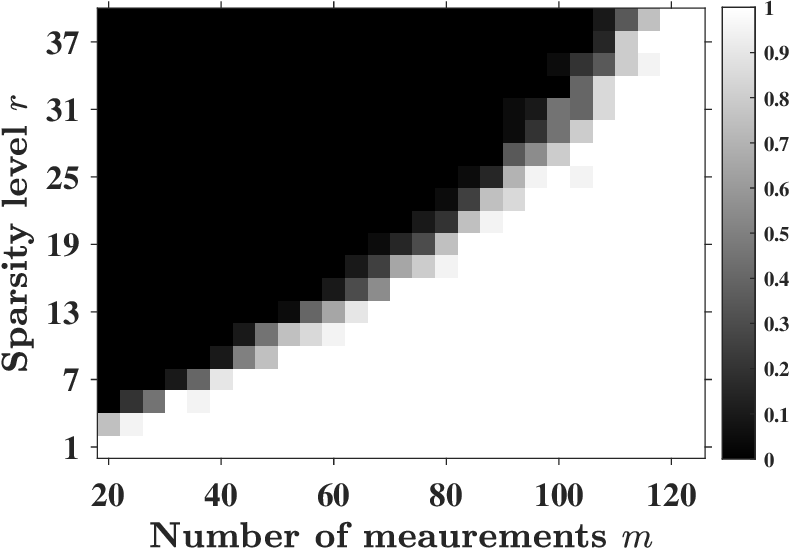}}
		\qquad
		\subfigure[ANM \cite{heckel2017generalized}]%
		{\includegraphics[width=.40\textwidth]{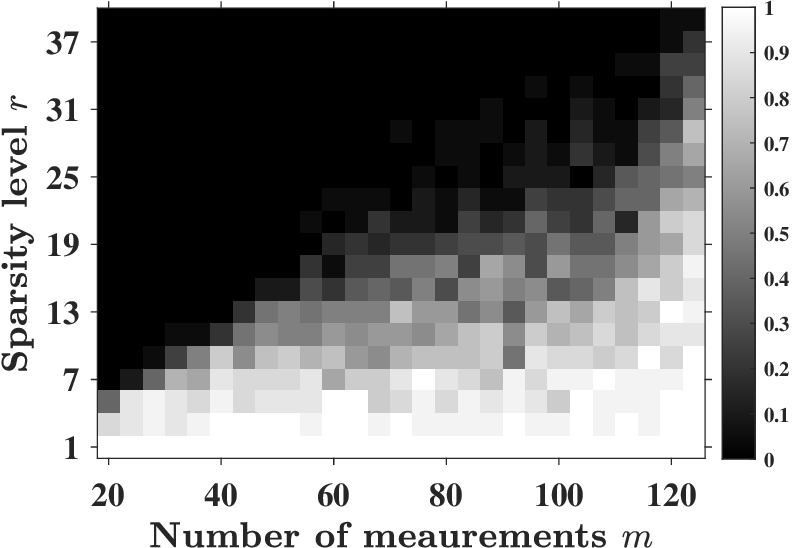}}
	\end{center}
	\caption{Phase transitions when $N=127$. White indicates success, while black indicates failure.} \label{fig:exp1}
\end{figure*}

\subsection{Robustness to Noise}
We evaluate the robustness of different methods in the presence of additive noise. We set $m=80$ and $r=5$, while the SNR varies from $5$ dB to $30$ dB. 
For both Hankel-NNM and ANM, we assume the noise energy $\twon{\m{\xi}}$ is known and impose the data-fidelity constraint $\twon{\m{A}\m{x}-\m{y}} \le \twon{\m{\xi}}$ in the optimization problem. For SIHT, we set $T_{max}=100$ and $K_{max} = 10$. For FIHT, we set $T_{max}=1000$.
Figure \ref{fig:noise} presents the relative RMSE and runtime under different SNR. We can see that SIHT and FIHT achieve nearly identical estimation accuracy and consistently outperform Hankel-NNM and ANM across all SNRs. Meanwhile, SIHT requires the least computational time among all compared methods, demonstrating its superior computational efficiency without sacrificing recovery accuracy.

\begin{figure*}[htb] 
	\centering
	\subfigure{\includegraphics[width=6.0cm]{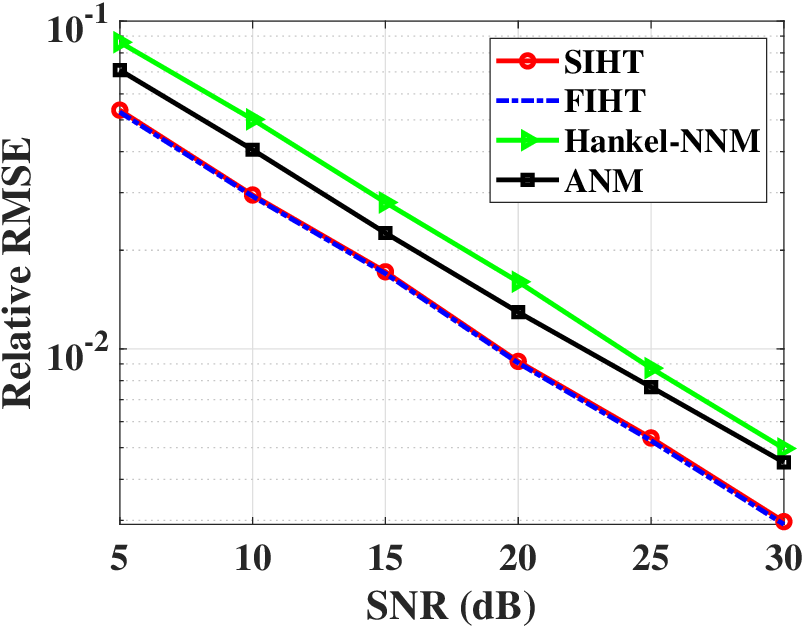}}
	\quad
	\subfigure{\includegraphics[width=6.0cm]{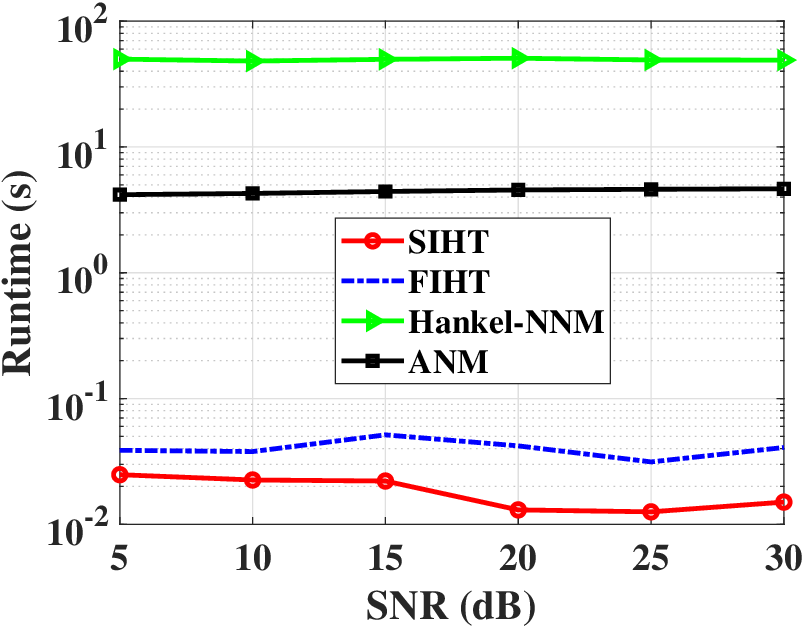}}
	\caption{Relative RMSE and runtime versus SNR. Here, \(N=127\), \(m=80\), and \(r=5\).}
	\label{fig:noise}
\end{figure*}

\subsection{Scalability to Large-Scale Problems}
Table \ref{tab:fixed_accuracy_time} reports the average runtime over $50$ Monte Carlo trials for achieving a target relative RMSE below $10^{-6}$ in the noiseless setting. As the problem dimension increases, SIHT consistently requires significantly less computational time than FIHT. These results demonstrate the superior scalability and computational efficiency of SIHT for large-scale spectrally sparse signal recovery. We do not include Hankel-NNM and ANM since their computational and memory costs become prohibitive as the problem size increases.

\begin{table}[t]
	\centering
	\caption{Average runtime comparison for achieving a target relative RMSE below $ 10^{-6}$.}
	\label{tab:fixed_accuracy_time}
	\begin{tabular}{c|c|c|cc}
		\hline
		$N$ & $m$ & $r$ & SIHT (s) & FIHT (s) \\
		\hline
		127  & 90   & 15  & {\bf 0.04}  & 0.17  \\
		527  & 400  & 50  & {\bf 0.74} & 1.81 \\
		1027 & 700  & 100 & {\bf 3.66} & 6.79 \\
		2027 & 1500 & 200 & {\bf 8.84} & 19.99 \\
		3027 & 2000 & 300 & {\bf 22.87} & 61.03 \\
		\hline
	\end{tabular}
\end{table}

\section{Proof of Theorem \ref{thm:recovery}} \label{sec:recovery}
Our proof proceeds in two stages: first, we prove that SIHT converges quadratically in a neighborhood of the ground truth; second, we show that the initialization falls within this neighborhood.

\subsection{Local Quadratic Convergence} \label{subsec:optimality}
Lemma \ref{lem:PTtPTt} establishes a restricted isometry property of $\cM$ on the tangent space $T$, which guarantees that \eqref{eq:subproblem} admits the unique solution and is crucial for proving Theorem \ref{thm:local}.

\begin{lemma}[Restricted Isometry of $\cM$ on Tangent Space $T$] \label{lem:PTtPTt}
	Let $T$ be the tangent space of a rank-$r$ matrix manifold. Let $\cM = \cP_T \sbra{ \cG \sbra{\m{B}^H\m{B} - \m{I}} \cG^* + \cI } \cP_T : \ T \to T$. Assume $\m{B}$ satisfies the LRIP of rank $2r$ with constant $\delta_{2r}<1$. 
	Then, $\mathcal{A}$ is positive definite and for all $\m{Z}\in T$, 
	\equ{
		\begin{split}
			\sbra{1 - \delta_{2r}} \frobn{\m{Z}} & \le \frobn{\cM \m{Z}} \le \sbra{1 + \delta_{2r}} \frobn{\m{Z}}, \\
			\frac{1}{1+\delta_{2r}} \frobn{\m{Z}} & \le \frobn{\cM^{-1} \m{Z}} \le \frac{1}{1-\delta_{2r}} \frobn{\m{Z}}. 
		\end{split} \label{eq:AinvZ}
	}
\end{lemma}
\begin{proof}
	The proof is deferred to Appendix \ref{subsec:PTtPTt}.
\end{proof}

\begin{theorem}[Local Quadratic Convergence] \label{thm:local}
	Assume $\m{B}$ satisfies the $3r$-LRIP with constant $\delta_{3r} < \frac{1}{2}$. Then the sequence $\{\m{X}_t\}$ in Algorithm \ref{alg:SIHT} satisfies
	\equs{
		\frobn{\m{X}_{t+1} - \cG\m{z}_\star}
		\le \frac{4}{\sigma_{\text{min}}\sbra{\cG\m{z}_\star }}  \frobn{\m{X}_t - \cG\m{z}_\star}^2 
		+ 8\sqrt{2r} \twon{\cG \m{B}^H \m{\xi}}, \quad \forall t \ge 0.
	}
	Further, if the initial estimate $\m{X}_0$ satisfies $\frobn{\m{X}_0 - \cG\m{z}_\star} \le \frac{1}{16} \sigma_{\min}(\cG\m{z}_\star)$, then
	\equs{
		\frobn{\m{X}_t - \cG\m{z}_\star} 
		\le 2^{-2^t} \sigma_{\min}(\cG\m{z}_\star) + 16 \sqrt{2r} \twon{\cG \m{B}^H \m{\xi}}, \quad \forall t \ge 0.
	}
\end{theorem}

\begin{proof}
	The proof is deferred to Appendix \ref{subsec:local}.
\end{proof}

\subsection{Initialization Error}
We provide the noise bound as follows.

\begin{lemma}[Concentration of Noise] \label{lem:noise}
	Assume the noise vector $\m{\xi}\in\bC^m$ has i.i.d.\ entries following $\cC\cN(0,\sigma^2/m)$.
	Assume that $\m{B}\in\bC^{m\times N}$ has i.i.d.\ entries following $\cC\cN(0,1/m)$ and is independent of $\m{\xi}$.
	Let $N>1$ and assume $m\ge(c_s+\log N)\log N$.
	Then, with probability at least $1-N^{-(c_s+\log N)}$, it holds that
	\equs{
		\twon{\cG\m{B}^H\m{\xi}}
		\le \sigma\sqrt{\frac{C(c_s+\log N)\log N}{m}},
	}
	where $C>0$ is a universal constant.
\end{lemma}

\begin{proof}
	The proof is deferred to Appendix \ref{append:noise}.
\end{proof}

Lemma \ref{lem:init1} establishes a upper bound on the one-step initialization error and extends \cite[Lemma 6]{wu2026lifted} to the noisy setting.

\begin{lemma}[One-Step Initialization Error] \label{lem:init1}
	Let $\m{X}_0 = \cT_r\sbra{ \cG \m{B}^H \m{y} }$ where $\m{y} = \m{B}\m{z}_\star + \m{\xi}$. Then,
	\equ{
		\frobn{\m{X}_0 - \cG\m{z}_{\star}} 
		\le 2 \delta_{2r} \frobn{\cG\m{z}_{\star}} 
		+ 2\sqrt{2r} \twon{\cG \m{B}^H \m{\xi}}.
		\label{eq:L0gz}
	}
\end{lemma}
\begin{proof}
	The proof is deferred to Appendix \ref{append:init1}.
\end{proof}

Lemma \ref{lem:init2} establishes an upper bound on the spectral norm error of the estimate by the resampled FIHT and extends \cite[Lemma 8]{wu2026lifted} to the noisy setting.

\begin{lemma}[Resampling Initialization Error] \label{lem:init2}
	Assume $\m{B}$ has i.i.d.\ entries following $\cC\cN\sbra{0,1/m}$ and $\m{\xi}$ has i.i.d.\ entries following $\cC\cN(0,\sigma^2/m)$. Let $\m{X}_0 = \tilde{\m{X}}_T$,  which is the output of Algorithm \ref{alg:Init-resamp}. Then,
	\equ{
		\twon{ \m{X}_0 - \cG\m{z}_\star} 
		\le 2^{-T-6} \sigma_{\min}(\cG\m{z}_\star) + \sigma \sqrt{ \frac{ C (c_s + \log N)\log N} {\tilde{m}} }. \label{eq:LTgz}
	}
	with probability at least $1 - (3+2T) N^{- (c_s + \log N)}$, where $C>0$ is a universal constant, provided $\tilde{m} \gtrsim \kappa^2 \sbra{ r + \frac{\sigma^2}{\sigma_{\max}^2\sbra{\cG\m{z}_\star}} } \max \lbra{c_s, \log N} \log N$.
\end{lemma}

\begin{proof}
	The proof is deferred to Appendix \ref{append:FIHT}.
\end{proof}

\subsection{Proof of Theorem \ref{thm:recovery}}
Since $\m{z}_t = \cG^*\m{X}_t$ and $\cG^*\cG = \cI$, we have
\equ{
	\twon{\m{z}_t-\m{z}_\star}
	= \twon{ \cG^*\sbra{ \m{X}_t - \cG\m{z}_\star } }
	\le \frobn{\m{X}_t - \cG\m{z}_\star}.
	\label{eq:zt+1}
}

	(i) One-step initialization: It follows from Lemma \ref{lem:noise} and Lemma \ref{lem:init1} that 
	\equs{
		\frobn{\m{X}_0 - \cG\m{z}_{\star}} 
		\le 2 \delta_{2r} \sqrt{r}\,\twon{\cG\m{z}_{\star}} 
		+ \sigma \sqrt{ \frac{ C r (c_s + \log N)\log N } {m} }.
	}
	To satisfy the condition $\frobn{\m{X}_0 - \cG\m{z}_\star}
	\le \tfrac{1}{16}\sigma_{\min}\sbra{\cG\m{z}_\star}$ in Theorem \ref{thm:local}, 
	we control the two terms separately. For the first term, it suffices that
	$
	\delta_{2r} \le \frac{1}{64\sqrt{r}\kappa}.
	$
	By Theorem \ref{thm:LRIP}, this holds provided
	$
	m \gtrsim C \kappa^2 r^2 \max \lbra{c_s, \log N} \log N,
	$
	which also ensures the conditions $\delta_{3r}<1/2$ in Theorem \ref{thm:local} and $\delta_1 < 1$ in Lemma \ref{lem:noise} are satisfied.
	For the second term, it suffices that
	$
	m \ge C \kappa^2 r \frac{\sigma^2}{\sigma_{\max}^2\sbra{\cG\m{z}_\star}} 
	\max\lbra{ c_s, \log N } \log N.
	$
	Combining the above two conditions, we obtain \eqref{eq:m_guarantee1}.
	Hence, by Theorem \ref{thm:local}, \eqref{eq:zt+1}, and Lemma \ref{lem:noise}, we obtain \eqref{eq:main_recovery}. 	
	
	(ii) Resampling initialization: We have
	\equs{
		\begin{split}
			\frobn{\m{X}_0 - \cG\m{z}_\star} 
			& \le \sqrt{2r} \twon{\m{X}_0  - \cG\m{z}_\star} \\
			& \le \sqrt{2r} \sbra{ 2^{-T - 6} \sigma_{\min}(\cG\m{z}_\star) + \sigma \sqrt{ \frac{ C (c_s + \log N)\log N} {\tilde{m}} } } \\
			& \le \frac{\sqrt{2}}{64} \sqrt{r} 2^{- \log_2 r } \sigma_{\text{min}} \sbra{\cG\m{z}_\star} + \frac{1}{32} \sigma_{\text{min}} \sbra{\cG\m{z}_\star} \\
			& \le \frac{1}{16} \sigma_{\text{min}} \sbra{\cG\m{z}_\star},
		\end{split} 
	}
	where the second inequality follows from Lemma \ref{lem:init2}, and the third one follows from $T \ge \lceil \log_2 r \rceil$ and the bound of $\tilde{m} $.
	Then by Theorem \ref{thm:local} and Lemma \ref{lem:noise}, we obtain \eqref{eq:main_recovery}.
	The requirement on $m$ in \eqref{eq:m_guarantee2} follows from the identity $m = (T+1)\tilde{m}$ together with the bound of $\tilde{m}$.
	
\section{Conclusion} \label{sec:conclusion}
We have proposed the first provably quadratically convergent nonconvex algorithm SIHT for spectral compressed sensing. SIHT achieves near-optimal sample complexity together with state-of-the-art computational and memory efficiency. Numerical experiments demonstrated that the superiority of SIHT. An interesting direction for future work is to develop second-order fast algorithms for low-rank Hankel matrix completion problem \cite{cai2019fast,chen2014robust}.

\appendix 

\section{Proof of Theorem \ref{thm:upperbound}} \label{append:upperbound}
	Consider the model \eqref{eq:P2}. By the optimality of $\widehat{\m{z}}_0$ and the feasibility of $\m{z}_\star$, we have $
	\twon{\m{B}\widehat{\m{z}}_0 - \m{y}} \le \twon{\m{B} \m{z}_\star - \m{y}}$. 
Substituting $\m{y}=\m{B}\m{z}_\star + \m{\xi}$ into it yields $
\twon{ \m{B} \sbra{ \widehat{\m{z}}_0 - \m{z}_\star} - \m{\xi} }  \le \twon{\m{\xi}}$.
It follows that
\equ{
	\twon{ \m{B} \sbra{ \widehat{\m{z}}_0 - \m{z}_\star} } 
	\le \twon{ \m{B} \sbra{ \widehat{\m{z}}_0 - \m{z}_\star} - \m{\xi} } + \twon{\m{\xi}} \le 2 \twon{\m{\xi}}. \label{eq:Bzz1}
}
On the other hand, using $\cG^*\cG=\cI$, we have
\equ{
	\twon{ \m{B} \sbra{ \widehat{\m{z}}_0 - \m{z}_\star} } 
	= \twon{ \m{B} \cG^* \sbra{ \cG \widehat{\m{z}}_0 - \cG\m{z}_\star} }  
	\ge  \sqrt{1-\delta_{2r}} \twon{ \widehat{\m{z}}_0 - \m{z}_\star }, \label{eq:Bzz2}
}
where the inequality follows from the LRIP and $\rank\sbra{\cG \widehat{\m{z}}_0 } = \rank\sbra{ \cG \m{z}_\star } = r$.
Combining \eqref{eq:Bzz1} and \eqref{eq:Bzz2} yields the upper bound of $\twon{ \widehat{\m{z}}_0 - \m{z}_\star }$.

Consider the model \eqref{eq:P5}. By the optimality of $\widehat{\m X}$ and the feasibility of $\cG\m z_\star$, then
\equs{
	\twon{\m B\cG^*\widehat{\m X}-\m y}^2
	+ \frobn{(\cI-\cG\cG^*)\widehat{\m X}}^2
	\le \twon{\m B\cG^*\cG\m z_\star-\m y}^2
	+ \frobn{(\cI-\cG\cG^*)\cG \m{z}_\star }^2.
}
Using $\m y=\m B\m z_\star+\m\xi$, we have $\twon{ \m B\cG^*(\widehat{\m X}-\cG\m z_\star)-\m\xi }^2
+ \frobn{ (\cI-\cG\cG^*)\widehat{\m X} }^2
\le \twon{ \m\xi }^2$.
It follows that $ \twon{ \m B\cG^*(\widehat{\m X}-\cG\m z_\star) - \m{\xi} } \le \twon{ \m\xi }$. The proof is completed by using the  argument as in \eqref{eq:Bzz1} and \eqref{eq:Bzz2}.

\section{Proof of Theorem \ref{thm:lowerbound}} \label{append:lowerbound}
The proof relies on the construction of two indistinguishable models. Consider the case $\m{B} = \m{I}$, under which $\m{B}$ satisfies LRIP and the model reduces to $\m{y} = \m{z}_\star + \m{\xi}$.
Define $r$-length all-ones vector $\m{1}_r$, and construct two $N$-length vectors $
\m{z}_1 = \frac{\varepsilon}{\sqrt{r}} \begin{bmatrix}
	\m{1}_r \\  \m{0}
\end{bmatrix}$ and $ 
\m{\xi}_1 = \frac{\varepsilon}{\sqrt{r}} \begin{bmatrix}
	\m{0} \\ \m{1}_r
\end{bmatrix}$.
Define the swapped pair $\m{z}_2 = \m{\xi}_1$ and $\m{\xi}_2 = \m{z}_1$. 
Note that $\rank\sbra{\cG\m{z}_1} = \rank\sbra{\cG\m{\xi}_1} = r$. Moreover, we have $\twon{\m{\xi}_1} = \twon{\m{\xi}_2} = \varepsilon$.
Thus $(\m{B},\m{z}_i,\m{\xi}_i)\in F_r(\varepsilon)$ for $i=1,2$.

Since $\m{z}_1 + \m{\xi}_1 = \m{z}_2 + \m{\xi}_2$, the two models yield the same observation $\m{y}$. Thus,
\equs{
	\begin{split}
		\inf_{\widehat{\m{z}}} \sup_{ \sbra{\m{B}, \m{z}_\star, \m{\xi}} \in F_r \sbra{\varepsilon} } \twon{\widehat{\m{z}} - \m{z}_\star}
		& \ge \inf_{\widehat{\m{z}}} \max \lbra{ \twon{\widehat{\m{z}} - \m{z}_1}, \twon{\widehat{\m{z}} - \m{z}_2} } \\
		& \ge \inf_{\widehat{\m{z}}} \frac{1}{2} \sbra{ \twon{\widehat{\m{z}} - \m{z}_1} + \twon{\widehat{\m{z}} - \m{z}_2} } \\
		& \ge \frac{1}{2} \twon{\m{z}_1 - \m{z}_2} 
		= \frac{1}{2} \twon{\m{z}_1 - \m{\xi}_1} 
		= \frac{\varepsilon}{\sqrt{2}},
	\end{split}
}
where the last equality holds because $r < \min\{n_1, n_2\}$ implies $2r \le N = n_1+n_2-1$ and thus the supports of $\m{z}_1$ and $\m{\xi}_1$ are disjoint.

\section{Proof of Lemma \ref{lem:PTtPTt}} \label{subsec:PTtPTt}
Note that $\cM$ is self-adjoint on $T$. To show the first line of \eqref{eq:AinvZ}, it is equivalent to show the spectrum of $\cM$ lies in $\mbra{ 1-\delta_{2r}, 1+\delta_{2r} }$.
Specifically, the eigenvalues of $\cM$ admit the variational characterization
\equ{
	\lambda_{\max}(\cM)
	:= \sup_{\frobn{\m{Z}}=1, \m{Z}\in T} \left\langle \cM \m{Z}, \m{Z} \right\rangle,
	\qquad
	\lambda_{\min}(\cM)
	:= \inf_{\frobn{\m{Z}}=1, \m{Z}\in T} \left\langle \cM \m{Z}, \m{Z} \right\rangle. \label{eq:lambda}
}
For any $\m{Z}\in T$, we have
\equ{
	\begin{split}
		\langle \cM \m{Z}, \m{Z} \rangle
		&= \left \langle \cP_T \sbra{ \cG \sbra{\m{B}^H\m{B} - \m{I}} \cG^* + \cI } \cP_T \m{Z}, \m{Z} \right \rangle \\
		&= \frobn{\m{Z}}^2 
		+ \left \langle (\m{B}^H\m{B} - \m{I}) \cG^* \m{Z}, \; \cG^* \m{Z} \right \rangle \\
		&= \frobn{\m{Z}}^2 + \twon{\m{B}\cG^* \m{Z}}^2 - \twon{\cG^* \m{Z}}^2,
	\end{split} \label{eq:<AZZ>}
}
where we use $\cP_T\m{Z}=\m{Z}$.
Since $\m{B}$ satisfies the LRIP of rank $2r$ and $\rank(\m{Z}) \le 2r$, we further have
\equs{
	(1-\delta_{2r}) \frobn{\m{Z}}^2
	\le \langle \cM \m{Z}, \m{Z} \rangle
	\le (1+\delta_{2r}) \frobn{\m{Z}}^2.
}
Substituting it into \eqref{eq:lambda} gives 
\equs{
	1-\delta_{2r} \le \lambda_{\min}(\cM)
	\le \lambda_{\max}(\cM)
	\le 1+\delta_{2r}. 
}
Since $\delta_{2r}<1$, it follows that $\lambda_{\min}(\cM)>0$, hence $\cM$ is positive definite on $T$. 

It follows that $\cM: T \to T$ is invertible.
For any $\m{X}\in T$, let $\m{Z} = \cM^{-1}\m{X} \in T$. Substituting $\m{Z}$ into the first line of \eqref{eq:AinvZ} gives the second line of \eqref{eq:AinvZ}.

\section{Proof of Theorem \ref{thm:local}} \label{subsec:local}

\begin{lemma}[{\cite[Lemma 10]{wu2026lifted}}] \label{lem:PTtI}
	Let $T$ be the tangent space of a rank-$r_1$ matrix manifold. Let $\m{Z}$ be another rank-$r_2$ matrix. Then, 
	\equs{
		\frobn{\cP_T \cG \sbra{ \m{B}^H \m{B} - \m{I} } \cG^*\m{Z} } \le \delta_{2r_1+r_2} \frobn{\m{Z}}.
	}
\end{lemma}	

\begin{lemma}[{\cite[Lemma 4.1]{wei2016guarantees}\cite[Lemma 12]{cai2025fast}}] \label{lem:wei}
	Let $T$ be the tangent space of the rank-$r$ matrix manifold at $\m{X}$. Let $\m{Z}$ be another rank-$r$ matrix. Then,
	\equs{
		\frobn{ \sbra{\cI - \cP_T} \m{Z} } \le \frac{\frobn{\m{X} - \m{Z}}^2}{\sigma_{\text{min}}\sbra{ \m{Z} }}, \qquad \twon{ \sbra{\cI - \cP_T}  \m{Z} } \le \frac{\twon{\m{X} - \m{Z} }^2}{\sigma_{\text{min}}\sbra{ \m{Z} }}.
	}
\end{lemma}

\begin{proof}[Proof of Theorem \ref{thm:local}] 
	We have
	\equ{
		\begin{split}
			\frobn{\m{X}_{t+1} - \cG\m{z}_\star} 
			& \le \frobn{\m{X}_{t+1} - \m{W}_{t+1}} + \frobn{\m{W}_{t+1} - \cG\m{z}_\star} \\ 
			& \le 2 \frobn{\m{W}_{t+1} - \cG\m{z}_\star} \\
			& \le 2 \frobn{ \m{W}_{t+1} - \cP_{T_t} \cG\m{z}_\star } + 2 \frobn{\sbra{\cI-\cP_{T_t}}\cG\m{z}_\star},
		\end{split} \label{eq:Lt+1}
	}
	where the second inequality holds since both $\m{X}_{t+1}$ and $\cG\m{z}_\star$ are rank-$r$ and $\m{X}_{t+1}$ is the best rank-$r$ approximation to $\m{W}_{t+1}$.
	
	Denote $\cM_t := \cP_{T_t}\sbra{\cG\sbra{\m{B}^H\m{B} - \m{I}}\cG^* + \cI } \cP_{T_t} $. It follows from the optimality condition \eqref{eq:linearEq} of problem \eqref{eq:subproblem} that
	$
		\cM_t \sbra{\m{W}_{t+1}} = \cP_{T_t} \cG \m{B}^H \m{y}.
	$
	As shown in Lemma \ref{lem:PTtPTt}, $\cM_t$ is invertible on $T_t$. Hence, we have
	\equ{
		\m{W}_{t+1} = \cM_t^{-1} \cP_{T_t} \cG \m{B}^H \m{y}. \label{eq:Wt+1}
	}
	Moreover, we have the decomposition
	\equ{
		\begin{split}
			\cP_{T_t} \cG \m{B}^H \m{y} 
			& = \cP_{T_t} \cG \m{B}^H \sbra{ \m{B}\m{z}_\star + \m{\xi} } \\
			& = \cM_t \cP_{T_t} \cG\m{z}_\star 
			+ \cP_{T_t} \cG \sbra{\m{B}^H\m{B} - \m{I}} \cG^*  \sbra{ \cI - \cP_{T_t} } \cG\m{z}_\star 
			+ \cP_{T_t} \cG \m{B}^H \m{\xi}.
		\end{split} \label{eq:PTtgBHy}
	}
	
	By substituting \eqref{eq:PTtgBHy} into \eqref{eq:Wt+1}, and subsequently into \eqref{eq:Lt+1}, we obtain
	\equ{
		\begin{split}
			\frobn{\m{X}_{t+1} - \cG\m{z}_\star}
			& \le 2 \frobn{ \cM_t^{-1} \mbra{ \cP_{T_t} \cG \sbra{\m{B}^H\m{B} - \m{I}} \cG^*  \sbra{ \cI - \cP_{T_t} } \cG\m{z}_\star + \cP_{T_t} \cG \m{B}^H \m{\xi} }  }  \\
			& \quad + 2 \frobn{\sbra{\cI-\cP_{T_t}}\cG\m{z}_\star} \\
			& \le \frac{2}{1-\delta_{2r}} \sbra{ \frobn{ \cP_{T_t} \cG \sbra{\m{B}^H\m{B} - \m{I}} \cG^*  \sbra{ \cI - \cP_{T_t}} \cG\m{z}_\star} + \frobn{\cP_{T_t} \cG \m{B}^H \m{\xi}} } \\
			& \quad  + 2 \frobn{\sbra{\cI-\cP_{T_t}}\cG\m{z}_\star} \\
			& \le \frac{2}{1-\delta_{2r}} \sbra{ \delta_{3r} \frobn{ \sbra{ \cI - \cP_{T_t} } \cG\m{z}_\star} + \frobn{ \cP_{T_t} \cG \m{B}^H \m{\xi}} } \\
			& \quad + 2 \frobn{\sbra{\cI-\cP_{T_t}}\cG\m{z}_\star} \\
			& \le \sbra{ \frac{2\delta_{3r}}{1-\delta_{2r}} + 2}  \frac{ \frobn{\m{X}_t - \cG\m{z}_\star}^2 }{\sigma_{\text{min}}\sbra{\cG\m{z}_\star } } 
			+ \frac{4\sqrt{2r}}{1-\delta_{2r}} \twon{\cG \m{B}^H \m{\xi}}.
		\end{split}  \label{eq:12}
	}
	Here the second inequality follows from Lemma \ref{lem:PTtPTt}, the third one follows from Lemma \ref{lem:PTtI} and $\rank\sbra{(\cI-\cP_{T_t})\cG\m{z}_\star}  = \rank \sbra{ \sbra{\m{I}-\m{U}_t\m{U}_t^H } \cG\m{z}_\star \sbra{\m{I}-\m{V}_t\m{V}_t^H } } \le r$, and the last one follows from Lemma \ref{lem:wei} and
	\equs{
		\frobn{\cP_{T_t} \cG \m{B}^H \m{\xi}} 
		\le \sqrt{2r} \, \twon{\cP_{T_t} \cG \m{B}^H \m{\xi}}
		\le 2\sqrt{2r} \, \twon{\cG \m{B}^H \m{\xi}},
	}
	since $\rank(\cP_{T_t}\m{X})\le 2r$ and
	$\twon{\cP_{T_t} \m{X} } = \twon{ \m{U}_t\m{U}_t^H \m{X} \sbra{ \m{I} - \m{V}_t\m{V}^H_t } + \m{X} \m{V}_t\m{V}^H_t } \le 2 \twon{\m{X}}$ for any $\m{X}$.
	
	Denote $e_t := \frobn{\m{X}_t - \cG\m{z}_\star}$, $a := \frac{4}{\sigma_{\min}(\cG\m{z}_\star)}$, and $b := 8\sqrt{2r} \twon{\cG \m{B}^H \m{\xi}}$. It follows from \eqref{eq:12} and $\delta_{2r}\le \delta_{3r} < \frac{1}{2}$ that
	$
		e_{t+1} \le a e_t^2 + b
	$.
	
	If $e_t \ge \sqrt{b/a}$, then $a e_t^2 \ge b$, and hence $e_{t+1} \le 2a e_t^2$. By induction, we have
	$
	e_t \le \frac{1}{2a} \sbra{2a e_0}^{2^t}
	$.
	Under the initialization condition $e_0 \le \frac{1}{16} \sigma_{\min}(\cG\m{z}_\star) = \frac{1}{4a}$, we further obtain
	$
	e_t \le \frac{1}{2a} 2^{-2^t}.
	$
	
	As $e_t$ decreases, eventually $e_t < \sqrt{b/a}$, equivalently, $a e_t^2 < b$. This yields 
	$e_{t+1} \le 2b$. 
	Consequently, combining the two regimes, we conclude that
	$e_t \le \frac{1}{2a} 2^{-2^t} + 2b$ for all $ t \ge 0$.
\end{proof}

\section{Proof of Lemma \ref{lem:noise}} \label{append:noise}
	Since $\m{B}$ has i.i.d.\ $\cC\cN(0,1/m)$ entries and is independent of $\m{\xi}$, conditioning on $\m{\xi}$ gives $
		\m{B}^H\m{\xi}\mid\m{\xi}
		\sim \cC\cN\left(
		\m{0},
		\frac{\twon{\m{\xi}}^2}{m}\m{I}_N
		\right)$.
	Consequently, conditionally on $\m{\xi}$, the random matrix $\cG\m{B}^H\m{\xi}$ has the same distribution as $\frac{\twon{\m{\xi}}}{\sqrt{m}}\cG\m{g}$, where $\m{g}\sim\cC\cN(\m{0},\m{I}_N)$ is independent of $\m{\xi}$.
	
	The Gaussian Hankel norm bound in \cite[Lemma 4]{wu2026lifted} gives
	\equs{
		\bE\twon{\cG\m{g}}
		\le 2\sqrt{(c_s+\log N)\log N}.
	}
	Moreover, for any $\m{g}_1,\m{g}_2\in\bC^N$,
	\equs{
		\abs{\twon{\cG\m{g}_1}-\twon{\cG\m{g}_2}}
		\le \twon{\cG(\m{g}_1-\m{g}_2)}
		\le \frobn{\cG(\m{g}_1-\m{g}_2)}
		= \twon{\m{g}_1-\m{g}_2}.
	}
	Thus, $\m{g}\mapsto\twon{\cG\m{g}}$ is $1$-Lipschitz.
	Applying Gaussian concentration \cite[Theorem 5.2.2]{vershynin2018high} to the real and imaginary coordinates yields
	\equs{
		\bP\left\{
		\twon{\cG\m{g}}
		>2\sqrt{(c_s+\log N)\log N}+C_0\sqrt{t}
		\right\}
		\le 2\exp(-t),
		\qquad t\ge0,
	}
	where $C_0\ge1$ is a universal constant.

	Taking $t=(c_s+\log N)\log N+\log 4$ and using $(c_s+\log N)\log N\ge\log 2$, we obtain
	\equ{
		\bP\left\{
		\twon{\cG\m{B}^H\m{\xi}}
		>4C_0\twon{\m{\xi}}
		\sqrt{\frac{(c_s+\log N)\log N}{m}}
		\,\middle|\,\m{\xi}
		\right\}
		\le \frac12 N^{-(c_s+\log N)}.
		\label{eq:gBxi}
	}
	
	Similarly, Gaussian concentration for the Euclidean norm, together with $\bE\twon{\m{\xi}}\le\sigma$, gives
	\equs{
		\bP\left\{
		\twon{\m{\xi}}
		>\sigma\left(1+C_0\sqrt{\frac{t}{m}}\right)
		\right\}
		\le 2\exp(-t),
		\qquad t\ge0.
	}
	Since $m\ge(c_s+\log N)\log N\ge\log 2$, taking $t=(c_s+\log N)\log N+\log 4$ yields
	\equ{
		\bP\lbra{\twon{\m{\xi}}>3C_0\sigma}
		\le \frac12 N^{-(c_s+\log N)}.
		\label{eq:noise_energy}
	}
	
	Integrating \eqref{eq:gBxi} with respect to $\m{\xi}$ and using \eqref{eq:noise_energy}, we conclude that
	\equs{
		\begin{split}
			& \bP\left\{
			\twon{\cG\m{B}^H\m{\xi}}
			>12C_0^2\sigma
			\sqrt{\frac{(c_s+\log N)\log N}{m}}
			\right\} \\
			& \le \bP\lbra{\twon{\m{\xi}}>3C_0\sigma} 
			+\bE_{\m{\xi}}\bP\left\{
			\twon{\cG\m{B}^H\m{\xi}}
			>4C_0\twon{\m{\xi}}
			\sqrt{\frac{(c_s+\log N)\log N}{m}}
			\,\middle|\,\m{\xi}
			\right\} \\
			& \le N^{-(c_s+\log N)}.
		\end{split}
	}
	This proves the claim with the universal constant $C=144C_0^4$.

\section{Proof of Lemma \ref{lem:init1}} \label{append:init1}
Let $\m{Q}\in \bC^{n_1 \times 2r}$ be an unitary matrix spanning the column subspaces of $\m{X}_0$ and $\cG\m{z}_\star$ and $
\cP_{\m Q}(\m X):=\m Q\m Q^H\m X$
be the orthogonal projection onto the column space of $\m Q$. It is shown in the proof of \cite[Lemma 6]{wu2026lifted} that 
$\frobn{\m{X}_0 - \cG\m{z}_{\star}} 
	\le 2 \frobn{ \cP_{\m{Q}} \sbra{\cG \m{B}^H \m{y}} - \cG\m{z}_{\star}}$.
Substituting $\m{y} = \m{B}\m{z}_\star + \m{\xi}$ into it yields
\equs{
	\begin{split}
		\frobn{\m{X}_0 - \cG\m{z}_{\star}} 
		& \le 2 \frobn{ \cP_{\m{Q}} \cG \sbra{ \m{B}^H \m{B} - \m{I} } \cG^* \cP_{\m{Q}} \sbra{ \cG \m{z}_{\star} } + \cP_{\m{Q}} \sbra{\cG \m{B}^H \m{\xi}} } \\
		& \le 2 \frobn{ \cP_{\m{Q}} \cG \sbra{ \m{B}^H \m{B} - \m{I} } \cG^* \cP_{\m{Q}} \sbra{ \cG \m{z}_{\star} } } + 2 \sqrt{2r} \twon{\cG \m{B}^H \m{\xi}},
	\end{split}
}
where the first inequality holds since $\cG\m{z}_\star = \cP_{\m{Q}} \sbra{\cG\m{z}_\star}$, and the second one follows from $\rank\sbra{\cP_{\m Q} \sbra{ \m X } } \le 2r$ and $\twon{ \cP_{\m Q} \sbra{ \m X }  } = \twon{ \m Q\m Q^H\m X } \le \twon{ \m Q\m Q^H } \twon{\m{X}} = \twon{\m{X}}$ for any $\m{X}$.
Applying \cite[Lemma 12]{wu2026lifted}, we obtain \eqref{eq:L0gz}.

\section{Resampled FIHT Algorithm and Proof of Lemma \ref{lem:init2}} \label{append:FIHT}
We partition the measurement index set $\lbra{1,\ldots,m}$ into $T+1$ disjoint subsets $\Omega_0, \ldots, \Omega_T$, each containing $\tilde{m}$ samples. Then we have $m = \sbra{T+1}\tilde{m}$.
Define
\equs{
		\begin{split}
				& \m{B} := \mbra{\m{b}_1,\ldots,\m{b}_m}^H \in \bC^{m\times N}, 
				\qquad \m{B}_t := \sqrt{\frac{m}{\tilde{m}}} \mbra{\m{b}_i^H}_{i\in \Omega_t} \in \bC^{\tilde{m} \times N}, \\ 
				& \m{y}_t := \sqrt{\frac{m}{\tilde{m}}} \mbra{y_i}_{i\in \Omega_t} \in \bC^{\tilde{m}},
				\qquad \qquad \m{\xi}_t := \sqrt{\frac{m}{\tilde{m}}} \mbra{\xi_i}_{i\in \Omega_t} \in \bC^{\tilde{m}}.
			\end{split} 
	}
The resampled FIHT algorithm is reviewed in Algorithm \ref{alg:Init-resamp}, where $\tilde{T}_t$ is the tangent space of the rank-$r$ matrix manifold at $\tilde{\m{X}}_t$.

\begin{algorithm}
	\caption{Initialization via Resampled FIHT \cite{wu2026lifted}}
	\label{alg:Init-resamp}
	\begin{algorithmic}
		\REQUIRE $\m{y}$, $\m{B}$, $r$ \\
		$\tilde{\m{X}}_0 = \cT_r\sbra{ \cG \m{B}_0^H \m{y}_0 }$
		\FOR{$t = 0,1,...,T-1$}
		\STATE $\tilde{\m{z}}_t = \cG^* \tilde{\m{X}}_t$
		\STATE $\tilde{\m{W}}_t = \cP_{\tilde{T}_t} \cG \sbra{ \tilde{\m{z}}_t - \m{B}_{t+1}^H \sbra{\m{B}_{t+1} \tilde{\m{z}}_t - \m{y}_{t+1} } }$
		\STATE $\tilde{\m{X}}_{t+1} = \cT_r ( \tilde{\m{W}}_t )$ \
		\ENDFOR
		\ENSURE $\tilde{\m{X}}_T$
	\end{algorithmic}
\end{algorithm}

\begin{proof}[Proof of Lemma \ref{lem:init2}]
It follows from $\m{y}=\m{B}\m{z}_\star+\m{\xi}$ that $\m{y}_t=\m{B}_t\m{z}_\star+\m{\xi}_t$. The entries of $\m{B}_t$ and $\m{\xi}_t$ are i.i.d.\ $\cC\cN(0,1/\tilde{m})$ and $\cC\cN(0,\sigma^2/\tilde{m})$, respectively. In Algorithm \ref{alg:Init-resamp}, the initialization error can be bounded as
\equs{
	\begin{split}
		\twon{\tilde{\m{X}}_0 - \cG\m{z}_\star}
		& \le 2 \twon{\cG\m{B}_0^H\m{y}_0 - \cG\m{z}_\star} \\
		& \le 2 \twon{\cG\m{B}_0^H\m{B}_0 \m{z}_\star - \cG\m{z}_\star} + 2 \twon{\cG\m{B}_0^H\m{\xi}_0 } \\
		& \le \frac{1}{128} \sigma_{\text{min}} \sbra{\cG \m{z}_\star} + \sigma \sqrt{ \frac{ C (c_s + \log N)\log N} {\tilde{m}} } \\
		& \le \frac{1}{64} \sigma_{\text{min}} \sbra{\cG \m{z}_\star},
	\end{split}
}
where the third inequality follows from \cite[Proof of Lemma 7]{wu2026lifted} as well as Lemma \ref{lem:noise} and holds with probability at least $1 - 3 N^{- (c_s + \log N)}$, and the last one holds under the condition of $\tilde{m}$.
For the iteration, we have
\equs{
	\begin{split}
		& \twon{\tilde{\m{X}}_{t+1} - \cG\m{z}_\star} \\
		& \le 2 \twon{ \cP_{\tilde{T}_t} \cG \mbra{ \tilde{\m{z}}_t - \m{B}^H_{t+1} \sbra{\m{B}_{t+1} \tilde{\m{z}}_t - \m{y}_{t+1}} } - \cG\m{z}_\star } \\
		& \le 2 \twon{ \cP_{\tilde{T}_t} \cG \sbra{ \cG^* \tilde{\m{X}}_t -  \m{B}_{t+1}^H \m{B}_{t+1} \cG^* \tilde{\m{X}}_t +  \m{B}_{t+1}^H \m{B}_{t+1} \m{z}_\star } - \cG \m{z}_\star }  + 2 \twon{ \cP_{\tilde{T}_t} \cG \m{B}_{t+1}^H \m{\xi}_{t+1} } \\ 
		& \le \sbra{ \frac{1}{8} + \frac{2}{\sigma_{\text{min}} \sbra{\cG\m{z}_\star}} \twon{ \tilde{\m{X}}_t - \cG\m{z}_\star } } \twon{ \tilde{\m{X}}_t - \cG\m{z}_\star } + 4 \sigma \sqrt{ \frac{ C (c_s + \log N)\log N} {\tilde{m}} },
	\end{split}
}
where the last inequality follows from \cite[Proof of Lemma 8]{wu2026lifted} as well as Lemma \ref{lem:noise} and holds with probability at least $1 - 2 N^{- (c_s + \log N)}$. The proof is completed by combining the above results.
\end{proof}

\bibliographystyle{siamplain}
\bibliography{SIHT}
\end{document}